\documentclass[letterpaper,12pt]{article}
\usepackage{diagbox}
\usepackage{latexsym,amssymb,amsmath, bm}
\usepackage{mathtools}
\usepackage{amsthm}
\usepackage{lipsum,graphicx,multicol}
\usepackage{setspace,color}
\usepackage{tcolorbox}
\usepackage{subcaption}
\usepackage{float}
\usepackage{booktabs}
\usepackage[normalem]{ulem}
\usepackage{authblk}
\usepackage{leftidx}
\usepackage{dsfont}
\usepackage{geometry}
\usepackage{bbm}
\usepackage[title,titletoc]{appendix}
\usepackage{actuarialangle}
\usepackage{lscape}
\usepackage{afterpage}
\usepackage{relsize}
\usepackage[colorlinks=true,allcolors=blue, linkcolor=red]{hyperref}
\usepackage{natbib}
\allowdisplaybreaks
\usepackage{lineno}
\usepackage{accents}

\def \P {\mathbb{P}}

\def \oo {\bm{0}}

\def \XX {\bm{X}}
\def \xx {\bm{x}}
\def \ww {\bm{w}}

\def \eqd {\buildrel d \over =}

\newcommand{\Fbar}{\,\overline{\kern-0.175em F\kern-0em}}

\def \E {\mathbb{E}}

\def \d {\mathrm{d}}

\def \aa {\boldsymbol{a}}
\def \bb {\boldsymbol{b}}

\def \ww {\boldsymbol{w}}
\def \xx {\boldsymbol{x}}

\def \NN {\mathbb{N}}

\def \RR {\mathbb{R}}

\def \XX {\boldsymbol{X}}
\def \YY {\boldsymbol{Y}}
\def \ZZ {\boldsymbol{Z}}

\def\VaR{{\mathrm{VaR}}}
\def\VaRplus{\mathrm{VaR}^{+}}

\numberwithin{equation}{section} 
\newtheorem{theorem}{Theorem}[section]
\newtheorem{definition}[theorem]{Definition}
\newtheorem{assumption}[theorem]{Assumption}

\newtheorem{example}[theorem]{Example}
\newtheorem{remark}[theorem]{Remark}
\newtheorem{proposition}[theorem]{Proposition}
\newtheorem{lemma}[theorem]{Lemma}
\newtheorem{corollary}[theorem]{Corollary}


\definecolor{darkread}{rgb}{0.7, 0, 0}

\begin{document}
\onehalfspacing
\title{\Large\textbf{{VaR at Its Extremes: Impossibilities and Conditions for One-Sided Random Variables}}}

\author[]{\small Nawaf Mohammed \thanks{\url{nawaf.mohammed.ac@gmail.com}}}

%\affil[]{\footnotesize Department of Mathematics and Statistics, York University, Toronto, Ontario, Canada}
\date{}
\maketitle
\vspace{-1cm}
\begin{abstract}
\noindent
Value-at-Risk (VaR) may reward diversification at some probability levels and penalize it at others. We study the two extremal regimes in which one effect prevails at every level simultaneously, for sums of one-sided risks. For risks supported on \([0,\infty)\), VaR sub-additivity at all levels, exact VaR additivity and co-monotonicity are equivalent without any integrability assumption; for such risks this removes the finite-mean hypothesis of Imamura and Kato (2026), which cannot be dropped for two-sided risks. Global super-additivity, in turn, is equivalent to first-order stochastic dominance of the aggregate over its co-monotonic counterpart. We factorize this condition through the maximum of independent copies of the margins into \emph{negative simplex dependence} (NSD), a condition on the joint law that is strictly weaker than negative lower orthant dependence, and \emph{simplex dominance} (SD) of the aggregator \(\Phi(x_1,\dots,x_n)=\sum_{i=1}^n x_i\log F_{X_i}(x_i)\), which for continuous margins holds whenever each reciprocal risk \(1/X_i\) has a decreasing failure rate on average. The criterion admits heterogeneous margins, yields a lower bound on the diversification penalty that depends on the margins only, is exact for unit Fr\'echet margins, and accommodates positively and tail-dependent risks such as common-shock models. Independently of any factorization, a vector that is VaR super-additive but not VaR additive has at least two infinite-mean components, and this bound is sharp. A reflection principle extends all results to finite endpoints: lower-bounded risks are never strictly sub-additive, upper-bounded risks are never strictly super-additive, and compactly supported risks are neither.
\vspace*{1cm}
~


\end{abstract}

\newpage
\section{Introduction}
\label{sec:intro}

\subsection{Motivation and related literature}
\label{subsec:motivation}
The Value-at-Risk (VaR) \citep{Linsmeier2000} -- the smallest capital for which the probability of a larger loss does not exceed a prescribed level -- remains the benchmark risk measure in insurance solvency regulation \citep{McNeil2015,Embrechts2014}. Its best-known weakness is the failure of sub-additivity, the axiom that encodes diversification in the theory of coherent risk measures \citep{Artzner1999} and that motivated the expected shortfall as a coherent alternative \citep{Acerbi.2002,Tasche2002,Wang2021}. VaR nevertheless retains a structural advantage: it is defined for every loss distribution, whereas expected shortfall requires integrability. For the infinite-mean models that arise in operational risk \citep{Jobst2007} and in catastrophe and liability lines \citep{Embrechts1997,Embrechts1999}, expected shortfall is identically infinite while VaR still ranks exposures -- and it is precisely there that an aggregate can be riskier than the sum of its parts \citep{Ibragimov2009,Ibragimov2008}.

Much of the literature treats these effects \emph{locally}, either asymptotically as the level tends to one \citep{Danielsson2013,Embrechts2008,Embrechts2009b,Zhu2023} or at a fixed regulatory level. We ask the \emph{global} question: when does one effect hold at \emph{every} level simultaneously? The question matters in practice because solvency capital, economic capital and pricing loadings are set at different confidence levels, and because its answer decides whether pooling or exchanging such risks can be beneficial at all \citep{Chen2025b,Mueller2025}. For super-additivity this question has recently been studied through first-order stochastic dominance \citep{Shaked2007}, beginning with \citet{Ibragimov2009} and continuing with \citet{Chen2025,Chen2026,Arab2025,Mueller2025,Chen2025a,Vincent2025,Zeng2025}; for sub-additivity, \citet{Imamura2025} showed that, for integrable risks, it forces co-monotonicity.

Global super-additivity may seem implausible, since co-monotonicity makes VaR exactly additive \citep{Dhaene2002,Denuit2005}. But co-monotonicity does not maximize \(\VaR_p[S]\) over a Fr\'echet class: the worst-case value generally exceeds the co-monotonic one \citep{Ruschendorf2013,Puccetti2013,Embrechts2014}, and it is in this gap that global super-additivity lives.

\subsection{Contributions}
\label{subsec:contributions}
The contributions of this paper are fourfold.

First, we prove an impossibility theorem (Theorem~\ref{thm:VaRsubadditive}): for risks supported on \([0,\infty)\), VaR sub-additivity at all levels, VaR additivity and co-monotonicity are equivalent, with no moment condition. This extends \citet{Imamura2025} to infinite-mean risks -- exactly the risks for which super-additivity occurs -- so that both extremes can be treated within one framework. One-sidedness is essential here: for two-sided risks the finite-mean hypothesis of \citet{Imamura2025} cannot be dropped (Remark~\ref{rm:role-of-one-sidedness}).

Second, we give a sufficient criterion for global VaR super-additivity (Theorem~\ref{thm:Varsuperadditivity}). Super-additivity is equivalent to first-order stochastic dominance of the aggregate over its co-monotonic counterpart, and to one inequality that entangles dependence and margins (Lemma~\ref{lm:aggregationcriterion}). We factorize it through the maximum of independent copies of the margins, which yields \emph{negative simplex dependence} (NSD) of the joint law and \emph{simplex dominance} (SD) of the aggregator \(\Phi\) in \eqref{eq:Phiexpression}. NSD says that the aggregate stochastically dominates this maximum and is strictly weaker than negative lower orthant dependence (NLOD) \citep{Block1982a,Joe1997}. For continuous margins, SD holds whenever each reciprocal risk \(1/X_i\) has a decreasing failure rate on average; for strictly increasing margins this means that each \(X_i\) is a star-shaped transformation of a unit Fr\'echet variable (Remark~\ref{rm:starFrechet}). The factorization also yields a lower bound on the diversification penalty that depends on the margins only (Corollary~\ref{coro:penalty-bound}). It is lossless for unit Fr\'echet margins, for which super-additivity is equivalent to NSD (Proposition~\ref{prop:frechet-boundary}).

Heterogeneous margins have previously been treated by comparing the weighted average with the mixture of the margins, for NLOD super-Fr\'echet risks by \citet{Chen2026} and for independent completely subscalable risks by \citet{Vincent2025}. Our criterion differs on both sides. NSD is strictly weaker than NLOD and accommodates positive dependence (Example~\ref{example:NSDnotNLOD}), including the upper tail dependence induced by a catastrophe shock common to all lines (Example~\ref{example:commonshock}). SD is a non-separable condition whose separable form replaces the convex transformations of the super-Fr\'echet class by star-shaped ones. The criterion recovers the super-additivity results of \citet{Chen2026} (Remark~\ref{rm:ChenShneer}) and is stable under star-shaped transformations (Proposition~\ref{prop:VaR_X_convex}), and explicit examples show that neither condition can be dropped and that neither is necessary.

Third, we determine how heavy the margins of a super-additive vector must be, independently of any factorization (Theorem~\ref{thm:intrinsic-heavy}). A vector that is VaR super-additive but not VaR additive has at least two infinite-mean components, and either two components with \(\limsup_{t\to\infty}t\,\Fbar_{X_i}(t)>0\) or one with \(\limsup_{t\to\infty}t\,\Fbar_{X_i}(t)=\infty\). The bound is sharp; its second alternative lies beyond the reach of SD and is attained by factorizing through the lower Fr\'echet bound instead of the independence copula (Proposition~\ref{prop:Wcriterion}).

Fourth, a reflection principle (Section~\ref{sec:Generalizations}) extends all results to arbitrary finite lower or upper endpoints, as arise under franchise deductibles \citep{Klugman2019} or capped covers. Table~\ref{tab:roadmap} summarizes the resulting duality.

\begin{table}[htb!]
\centering
\small
\resizebox{\textwidth}{!}{%
\begin{tabular}{@{}lll@{}}
\toprule
Margins supported on & VaR sub-additivity at all levels & VaR super-additivity at all levels \\
\midrule
\([a_i,\infty)\), \(a_i>-\infty\)
 & \(\iff\) additivity \(\iff\) co-monotonicity
 & attainable; NSD \(+\) SD suffice \\
 & \small(Thm.~\ref{thm:VaRsubadditive}, Prop.~\ref{prop:XaXbImpossibility}(i))
 & \small(Thm.~\ref{thm:Varsuperadditivity}, Prop.~\ref{prop:VaR_super-extension}(i)) \\[2mm]
\((-\infty,b_i]\), \(b_i<\infty\)
 & attainable; reflected NSD \(+\) SD suffice
 & \(\iff\) additivity \(\iff\) co-monotonicity \\
 & \small(Prop.~\ref{prop:VaR_super-extension}(ii), Prop.~\ref{prop:VaR_super-extension-sufficient}(ii))
 & \small(Prop.~\ref{prop:XaXbImpossibility}(ii)) \\[2mm]
\([a_i,b_i]\), compact
 & \(\iff\) additivity \(\iff\) co-monotonicity
 & \(\iff\) additivity \(\iff\) co-monotonicity \\
\bottomrule
\end{tabular}}
\caption{Roadmap of the paper. Each cell records what can happen when the stated additivity relation is required to hold at \emph{every} probability level \(p\in(0,1)\).}
\label{tab:roadmap}
\end{table}

\subsection{Setting, notation and preliminaries}
\label{subsec:setting}
Throughout, we work on an atomless probability space \((\Omega,\mathcal{F},\P)\) and consider random vectors \(\XX=(X_1,\dots,X_n)\), \(n\in\NN\) with \(n\ge 2\), whose components are non-degenerate random variables -- that is, none of them is almost surely constant -- representing, for example, asset prices or insurance losses. We focus in particular on their aggregate,
\[
S = \sum_{i=1}^n X_i.
\]

For a random variable or vector we write \(f\), \(F\) and \(\Fbar\) for its density, cumulative distribution function (CDF) and survival, or decumulative distribution, function (DDF), indexed by the variable, e.g.\ \(F_{\XX}\) and \(F_{X_i}\). No integrability is assumed unless stated. We set \(\log 0=-\infty\) and \(0\cdot\log 0:=0\).

The lower and upper endpoints of the support of \(X_i\) are written
\[
a_i = \sup\{x \in \RR : F_{X_i}(x) \le 0\},
\qquad
b_i = \inf\{x \in \RR : F_{X_i}(x) \ge 1\},
\qquad i \in \{1,\dots,n\},
\]
with the usual conventions \(\sup\emptyset=-\infty\) and \(\inf\emptyset=+\infty\). The symbols \(a_i\) and \(b_i\) are reserved for these endpoints throughout; in parametric families, \(\alpha_i\) denotes a shape and \(\theta_i\) a scale parameter. Sections~\ref{sec:VaR_subadditivity} and~\ref{sec:VaR_super-additivity} operate under the following normalization.

\begin{assumption}
\label{ass:zero-endpoint}
Each \(X_i\) has lower endpoint at the origin, \(a_i=0\) for all \(i\in\{1,\dots,n\}\); equivalently,
\[
\P(X_i\ge 0)=1
\qquad\text{and}\qquad
\P(X_i\le \varepsilon)>0
\quad \forall\, \varepsilon>0 ,
\qquad \forall\, i \in \{1,\dots,n\}.
\]
\end{assumption}

Section~\ref{sec:Generalizations} removes this normalization for arbitrary finite lower endpoints, so within that class it entails no loss of generality.

Finally, for any random variable \(Z\) and any confidence level \(p\in(0,1)\), we work with both generalized inverses of \(F_Z\) \citep{Embrechts2013}, namely the left-quantile
\[
\VaR_p[Z] = \inf\{z \in \RR : F_Z(z) \ge p\}
\]
and the right-quantile
\[
\VaRplus_p[Z] = \inf\{z \in \RR : F_Z(z) > p\} = \sup\{z\in\RR : F_Z(z)\le p\} .
\]
``VaR'' always means the left-quantile \(\VaR_p\); the right-quantile is needed only because it arises under reflection (Lemma~\ref{lm:VaRtoolkit}(iv)).

\begin{lemma}
\label{lm:VaRtoolkit}
Let \(Z\) be a random variable, let \(c\in\RR\), \(\lambda>0\) and \(p\in(0,1)\). Then:
\begin{itemize}
\item[(i)] \(p\mapsto \VaR_p[Z]\) is non-decreasing and left-continuous, \(p\mapsto \VaRplus_p[Z]\) is non-decreasing and right-continuous, and
\[
\VaR_p[Z]\;\le\;\VaRplus_p[Z]
=\lim_{q\downarrow p}\VaR_q[Z],
\qquad
\VaR_p[Z]=\lim_{q\uparrow p}\VaRplus_q[Z].
\]
Moreover \(\VaR_p[Z]=\VaRplus_p[Z]\) for all but at most countably many \(p\in(0,1)\).
\item[(ii)] {\rm (Translation and scale equivariance.)} \(\VaR_p[c+\lambda Z]=c+\lambda\,\VaR_p[Z]\), and likewise for \(\VaRplus_p\).
\item[(iii)] {\rm (Co-monotonic additivity.)} If \(\XX\) is co-monotonic in the sense of Definition~\ref{def:co-monotonicity}, then \(\VaR_p[S]=\sum_{i=1}^n\VaR_p[X_i]\) for every \(p\in(0,1)\).
\item[(iv)] {\rm (Reflection.)} \(\VaR_p[c-Z]=c-\VaRplus_{1-p}[Z]\) and \(\VaRplus_p[c-Z]=c-\VaR_{1-p}[Z]\).
\end{itemize}
\end{lemma}

\begin{proof}
Parts (i)--(iii) are standard; see \citet{Embrechts2013}, \citet[Appendix A.1]{McNeil2015} and \citet{Dhaene2002}. For (iv), \(F_{c-Z}(y)=1-F_Z\big((c-y)^-\big)\), so \(\VaR_p[c-Z]=c-\sup\{z:F_Z(z^-)\le 1-p\}\). Since \(F_Z(z^-)\le q\) holds exactly when \(z\le\VaRplus_q[Z]\), the supremum equals \(\VaRplus_{1-p}[Z]\). The second identity is proved in the same way.
\end{proof}

We compare \(\VaR_p[S]\) with \(\sum_{i=1}^n\VaR_p[X_i]\) at every level \(p\in(0,1)\), as arises when an insurer aggregates heterogeneous and dependent lines of business.
\begin{definition}
\label{def:VaRsubsuper}
We say that $\XX$ is 
{VaR sub-additive} %or simply {sub-additive}
(respectively, {VaR super-additive}) if
\begin{equation}
\label{eq:VaRsubsuper}
\VaR_p[S] \le (\ge) \; \sum_{i=1}^n\VaR_p[X_i], 
\quad \forall p \in (0,1).
\end{equation}
In particular, $\XX$ is called {VaR additive} if equality holds for all probability levels $p \in (0,1)$.
\end{definition}

Definition~\ref{def:VaRsubsuper} depends neither on the quantile convention nor on negligible sets of levels. This justifies numerical verification on a grid and the passage through right-quantiles in Section~\ref{sec:Generalizations}.

\begin{lemma}
\label{lm:VaRversion}
Let \(D\subseteq(0,1)\) be dense. The following are equivalent:
\begin{itemize}
\item[(a)] \(\VaR_p[S]\ \ge\ \sum_{i=1}^n \VaR_p[X_i]\) for all \(p\in(0,1)\), that is, \(\XX\) is VaR super-additive;
\item[(b)] \(\VaR_p[S]\ \ge\ \sum_{i=1}^n \VaR_p[X_i]\) for all \(p\in D\);
\item[(c)] \(\VaR_p[S]\ \ge\ \sum_{i=1}^n \VaR_p[X_i]\) for Lebesgue-almost every \(p\in(0,1)\);
\item[(d)] \(\VaRplus_p[S]\ \ge\ \sum_{i=1}^n \VaRplus_p[X_i]\) for all \(p\in(0,1)\).
\end{itemize}
The same equivalences hold with every inequality reversed, and hence also for VaR sub-additivity and for VaR additivity.
\end{lemma}

\begin{proof}
(a)\(\Rightarrow\)(b) and (a)\(\Rightarrow\)(c) are immediate. (b)\(\Rightarrow\)(a): fix \(p\) and take \(q_k\in D\) with \(q_k\uparrow p\); both sides of the inequality are left-continuous in the level (Lemma~\ref{lm:VaRtoolkit}(i)), so letting \(k\to\infty\) gives (a). (c)\(\Rightarrow\)(a): the set of levels at which the inequality holds has full Lebesgue measure, hence is dense, and the previous argument applies to it. (a)\(\Rightarrow\)(d): let \(q\downarrow p\) in (a) and use \(\VaRplus_p=\lim_{q\downarrow p}\VaR_q\). (d)\(\Rightarrow\)(a): let \(q\uparrow p\) in (d) and use \(\VaR_p=\lim_{q\uparrow p}\VaRplus_q\). Reversed inequalities are handled identically.
\end{proof}

\begin{remark}
\label{rm:positive-quantiles}
Under Assumption~\ref{ass:zero-endpoint}, if \(F_{X_i}\) is continuous then \(F_{X_i}(0)=0\) and \(F_{X_i}(\VaR_p[X_i])=p\); hence \(\VaR_p[X_i]>0\) and \(\sum_{i=1}^n\VaR_p[X_i]>0\) for every \(p\in(0,1)\). If moreover \(F_{X_i}\) is strictly increasing on \((0,b_i)\), then \(\VaR_p[X_i]=\VaRplus_p[X_i]\) for all \(p\); continuity alone does not suffice, since a flat stretch of \(F_{X_i}\) at height \(p\) separates the two quantiles.
\end{remark}

Section~\ref{sec:VaR_subadditivity} proves the impossibility theorem, Section~\ref{sec:VaR_super-additivity} develops the NSD--SD criterion and determines how heavy the margins of a super-additive vector must be, Section~\ref{sec:Generalizations} treats general endpoints, and Section~\ref{sec:conclusions} concludes. Longer proofs and computations are collected in the Appendix.
%%%%%%%%%%%%%%%%%%%%%%%%%%%%%%%%%%%%%%%%%%%%%%%%%%%%%%%%%%%%%%%%%%%%%%%%%%%%%%%%%%%%%%%%%%%%%%%%%%%%%
\section{VaR Sub-additivity}
\label{sec:VaR_subadditivity}
This section settles the upper-left cell of Table~\ref{tab:roadmap}. For integrable risks, \citet{Imamura2025} showed that VaR sub-additivity at all levels is equivalent to co-monotonicity. We show that the equivalence holds for all non-negative risks, including those with infinite means. The proof approximates \(\XX\) by bounded conditional versions \(\XX_k\), transfers sub-additivity to them, applies \citet{Imamura2025}, and passes to the limit. We first recall co-monotonicity \citep{Dhaene2002,Denuit2005}.
\begin{definition}
\label{def:co-monotonicity}
A random vector $\XX$ is co-monotonic if its joint CDF $F_{\XX}$ is the Fr\'echet upper bound 
\begin{equation*}
F_{\XX}(x_1,\dots,x_n)=\min\left\{F_{X_{1}}(x_1),\dots,F_{X_{n}}(x_n)\right\}.
\end{equation*}
\end{definition}

\begin{remark}
\label{rm:comonotonicity-facts}
We shall use three standard reformulations of Definition~\ref{def:co-monotonicity}, all established in \citet[Section 4]{Dhaene2002}; see also \citet[Chapter 4]{Denuit2005}. The vector $\XX$ is co-monotonic if and only if
\begin{itemize}
\item[(a)] there exists a random variable $U$, uniformly distributed on $(0,1)$, such that $X_i=\VaR_U[X_i]$ almost surely for every $i\in\{1,\dots,n\}$;
\item[(b)] the support of $\XX$ is a chain in $\RR^n$, that is, any two of its points are ordered componentwise;
\item[(c)] every bivariate margin $(X_i,X_j)$, $i\neq j$, is co-monotonic.
\end{itemize}
Representation (a) shows that co-monotonicity passes to sub-vectors, and (c) reduces its refutation to a pairwise check.
\end{remark}

The proof rests on the following device: conditioning on a lower-tail event of probability \(c\) rescales a CDF by \(c\) and caps it at one.
\begin{lemma}
\label{lm:truncation}
Let $Z$ be a random variable with CDF $F_Z$, and let $c\in(0,1]$. Define
\[
F_Z^{[c]}(z)=\min\left\{\dfrac{F_Z(z)}{c},\,1\right\},\qquad z\in\RR .
\]
Then $F_Z^{[c]}$ is a CDF and, for every $p\in(0,1)$,
\begin{equation}
\label{eq:truncation-quantile}
\inf\{z\in\RR : F_Z^{[c]}(z)\ge p\}=\VaR_{pc}[Z].
\end{equation}
Moreover, if $\widetilde F$ is any CDF satisfying $\widetilde F\le F_Z^{[c]}$ pointwise, then
\begin{equation}
\label{eq:truncation-comparison}
\inf\{z\in\RR : \widetilde F(z)\ge p\}\;\ge\;\VaR_{pc}[Z],\qquad \forall p\in(0,1).
\end{equation}
\end{lemma}
\begin{proof}
As the minimum of two non-decreasing and right-continuous functions, $F_Z^{[c]}$ is non-decreasing and right-continuous; moreover $F_Z^{[c]}(z)\to0$ as $z\to-\infty$, and $F_Z^{[c]}(z)\to 1$ as $z\to+\infty$ because $F_Z(z)\to1\ge c$. Hence $F_Z^{[c]}$ is a CDF.

Fix $p\in(0,1)$. Since $p<1$, we have $\min\{u,1\}\ge p$ if and only if $u\ge p$, for any $u\ge0$. Consequently
\[
\{z\in\RR: F_Z^{[c]}(z)\ge p\}
=\left\{z\in\RR: \dfrac{F_Z(z)}{c}\ge p\right\}
=\{z\in\RR: F_Z(z)\ge pc\},
\]
and taking infima yields \eqref{eq:truncation-quantile}, the right-hand side being well defined because $pc\in(0,1)$. Finally, if $\widetilde F\le F_Z^{[c]}$ pointwise then $\{z: \widetilde F(z)\ge p\}\subseteq\{z: F_Z^{[c]}(z)\ge p\}$, so the infimum over the former set is at least the infimum over the latter, which is \eqref{eq:truncation-comparison}.
\end{proof}

We state the main result of this section in its sharpest form.
\begin{theorem}
\label{thm:VaRsubadditive}
Let $\XX$ satisfy Assumption~\ref{ass:zero-endpoint}. Then the following statements are equivalent:
\begin{itemize}
\item[(i)] $\XX$ is VaR sub-additive;
\item[(ii)] $\XX$ is VaR additive;
\item[(iii)] $\XX$ is co-monotonic.
\end{itemize}
In particular, strict VaR sub-additivity is impossible: if $\VaR_{p}[S]<\sum_{i=1}^n\VaR_{p}[X_i]$ for even a single level $p\in(0,1)$, then $\VaR_{q}[S]>\sum_{i=1}^n\VaR_{q}[X_i]$ for some other level $q\in(0,1)$.
\end{theorem}

\begin{proof}
The implication (iii) $\Rightarrow$ (ii) is the co-monotonic additivity of quantiles, Lemma~\ref{lm:VaRtoolkit}(iii), and (ii) $\Rightarrow$ (i) is immediate from Definition~\ref{def:VaRsubsuper}. The final assertion follows from the equivalence, since a vector that is not VaR sub-additive violates the inequality in \eqref{eq:VaRsubsuper} at some level. It therefore remains to prove (i) $\Rightarrow$ (iii).

\medskip

Suppose that $\XX$ is VaR sub-additive, and let
\[
a_S=\sup\{s\in\RR: F_S(s)\le 0\}
\]
denote the lower endpoint of the aggregate $S$. By Assumption~\ref{ass:zero-endpoint} we have $S\ge0$ almost surely, so $a_S\ge0$; and since $S$ is almost surely finite, $F_S(s)\to1$ as $s\to\infty$, whence $a_S<\infty$. Fix any constant $k>a_S$ and write
\[
c_k:=F_S(k)=\P(S\le k)\in(0,1].
\]
Let $\XX_k=(X_{1,k},\dots,X_{n,k})$ denote the random vector distributed as $\XX$ conditionally on the event $\{S\le k\}$, that is,
\begin{equation}
\label{eq:conditional-vector}
\P\left(\XX_k\in B\right)=\dfrac{\P\left(\XX\in B,\ S\le k\right)}{c_k},
\qquad B\in\mathcal{B}(\RR^n),
\end{equation}
and set $S_k=\sum_{i=1}^n X_{i,k}$, so that $S_k$ is distributed as $S$ conditionally on $\{S\le k\}$.

\medskip
\noindent\textit{Step 1: $\XX_k$ is bounded, hence integrable.}
On the event $\{S\le k\}$, Assumption~\ref{ass:zero-endpoint} gives $X_j\ge0$ for every $j$, and therefore
\[
0\;\le\; X_i\;\le\;\sum_{j=1}^n X_j = S\;\le\; k,
\qquad \forall i\in\{1,\dots,n\}.
\]
Consequently $0\le X_{i,k}\le k$ almost surely, and in particular $\E[X_{i,k}]\le k<\infty$ for every $i$.

\medskip
\noindent\textit{Step 2: identification of the relevant distribution functions.}
Taking $B=\{\xx\in\RR^n : x_i\le x\}$ in \eqref{eq:conditional-vector} and using $\P(A\cap C)\le\min\{\P(A),\P(C)\}$, we obtain, for every $x\in\RR$,
\begin{equation}
\label{eq:margin-domination}
F_{X_{i,k}}(x)
=\dfrac{\P(X_i\le x,\ S\le k)}{c_k}
\;\le\;\dfrac{\min\left\{F_{X_i}(x),\,c_k\right\}}{c_k}
=\min\left\{\dfrac{F_{X_i}(x)}{c_k},\,1\right\}
=:F_{X_i}^{[c_k]}(x).
\end{equation}
For the aggregate, $\{S\le s\}\cap\{S\le k\}=\{S\le s\wedge k\}$ yields the exact identity
\begin{equation}
\label{eq:sum-normalisation}
F_{S_k}(s)=\dfrac{F_S(s\wedge k)}{c_k}
=\min\left\{\dfrac{F_S(s)}{c_k},\,1\right\},
\qquad \forall s\in\RR,
\end{equation}
the second equality holding because $F_S(s)\le c_k$ for $s<k$ and $F_S(s)\ge c_k$ for $s\ge k$.

\medskip
\noindent\textit{Step 3: $\XX_k$ inherits VaR sub-additivity.}
Fix $p\in(0,1)$ and note that $pc_k\in(0,1)$. Applying \eqref{eq:truncation-quantile} of Lemma~\ref{lm:truncation} to \eqref{eq:sum-normalisation} with $Z=S$ and $c=c_k$ gives
\[
\VaR_p[S_k]=\VaR_{pc_k}[S],
\]
while applying \eqref{eq:truncation-comparison} to \eqref{eq:margin-domination} with $Z=X_i$, $c=c_k$ and $\widetilde F=F_{X_{i,k}}$ gives
\[
\VaR_p[X_{i,k}]\;\ge\;\VaR_{pc_k}[X_i],\qquad \forall i\in\{1,\dots,n\}.
\]
Since $\XX$ is VaR sub-additive, the inequality in \eqref{eq:VaRsubsuper} may be invoked at the level $pc_k$, and chaining the three displays yields
\[
\VaR_p[S_k]
=\VaR_{pc_k}[S]
\;\le\;\sum_{i=1}^n\VaR_{pc_k}[X_i]
\;\le\;\sum_{i=1}^n\VaR_p[X_{i,k}],
\qquad \forall p\in(0,1).
\]
Hence $\XX_k$ is VaR sub-additive.

\medskip
\noindent\textit{Step 4: co-monotonicity of the approximations.}
By Step 1 each $X_{i,k}$ has a finite expectation, and by Step 3 the vector $\XX_k$ is VaR sub-additive at every confidence level. Theorem 1 of \citet{Imamura2025}, which is stated on an atomless probability space and involves \(\XX_k\) only through its law, therefore applies and yields that $\XX_k$ is co-monotonic, that is,
\begin{equation}
\label{eq:comonotone-truncated}
F_{\XX_k}(x_1,\dots,x_n)
=\min\left\{F_{X_{1,k}}(x_1),\dots,F_{X_{n,k}}(x_n)\right\},
\qquad \forall(x_1,\dots,x_n)\in\RR^n .
\end{equation}

\medskip
\noindent\textit{Step 5: passage to the limit.}
Let $(k_m)_{m\ge1}$ be any sequence with $k_m>a_S$ and $k_m\uparrow\infty$. Fix $\xx=(x_1,\dots,x_n)\in\RR^n$. The events $\{\XX\le\xx,\,S\le k_m\}$ increase to $\{\XX\le\xx,\,S<\infty\}$, which coincides with $\{\XX\le\xx\}$ up to a null set because $S$ is almost surely finite. Monotone convergence of probability measures therefore gives
\[
\P\left(\XX\le\xx,\ S\le k_m\right)\longrightarrow F_{\XX}(\xx),
\qquad
c_{k_m}=\P(S\le k_m)\longrightarrow 1,
\]
and hence, by \eqref{eq:conditional-vector},
\[
F_{\XX_{k_m}}(\xx)\longrightarrow F_{\XX}(\xx),
\qquad m\to\infty .
\]
The identical argument applied coordinatewise gives $F_{X_{i,k_m}}(x_i)\to F_{X_i}(x_i)$ for every $i$. Since the map $(u_1,\dots,u_n)\mapsto\min\{u_1,\dots,u_n\}$ is continuous, letting $m\to\infty$ in \eqref{eq:comonotone-truncated} produces
\[
F_{\XX}(x_1,\dots,x_n)=\min\left\{F_{X_1}(x_1),\dots,F_{X_n}(x_n)\right\},
\qquad \forall(x_1,\dots,x_n)\in\RR^n .
\]
Thus $\XX$ is co-monotonic in the sense of Definition~\ref{def:co-monotonicity}, which establishes (iii) and, through Lemma~\ref{lm:VaRtoolkit}(iii), returns (ii).
\end{proof}

\begin{remark}
\label{rm:role-of-one-sidedness}
The decisive use of Assumption~\ref{ass:zero-endpoint} is in Step~1: non-negativity gives \(X_i\le S\), which turns a bound on the aggregate into a bound on each component. This is what replaces the integrability hypothesis of \citet{Imamura2025}, and without either of them the conclusion fails. For independent standard Cauchy variables \(X_1,X_2\) one has \(X_1+X_2\eqd2X_1\), so that \(\VaR_p[X_1+X_2]=\VaR_p[X_1]+\VaR_p[X_2]\) at every level although \((X_1,X_2)\) is not co-monotonic. The finite-mean hypothesis of \citet{Imamura2025} therefore cannot be dropped for two-sided risks. The restriction \(k>a_S\) is genuinely needed: for \(\XX=(U,1-U)\) with \(U\) uniform on \((0,1)\), both margins have lower endpoint \(0\), yet \(S\equiv1\) and \(a_S=1\). Finally, by Lemma~\ref{lm:VaRversion}, the conclusion already follows if the sub-additivity inequality holds on a dense set of levels.
\end{remark}

Since co-monotonicity pins down the joint law once the margins are fixed, its failure is generic. The contrapositive of Theorem~\ref{thm:VaRsubadditive}, combined with left-continuity, then yields strict super-additivity on a set of levels of positive measure.
\begin{corollary}
\label{coro:genericsuperadditivity}
Let $\XX$ satisfy Assumption~\ref{ass:zero-endpoint} and suppose that $\XX$ is not co-monotonic. Then there exist $p^*\in(0,1)$ and $\delta>0$ such that
\[
\VaR_p[S]\;>\;\sum_{i=1}^n\VaR_p[X_i],
\qquad \forall p\in(p^*-\delta,\,p^*] .
\]
In particular, strict VaR super-additivity holds on a set of confidence levels of positive Lebesgue measure.
\end{corollary}
\begin{proof}
By Theorem~\ref{thm:VaRsubadditive}, a vector that is not co-monotonic is not VaR sub-additive, so the inequality in \eqref{eq:VaRsubsuper} fails at some level $p^*\in(0,1)$; that is, the function
\[
\Delta(p):=\VaR_p[S]-\sum_{i=1}^n\VaR_p[X_i]
\]
satisfies $\Delta(p^*)>0$. By Lemma~\ref{lm:VaRtoolkit}(i) the maps $p\mapsto\VaR_p[S]$ and $p\mapsto\VaR_p[X_i]$ are left-continuous on $(0,1)$, hence so is the finite combination $\Delta$. Left-continuity at $p^*$ supplies $\delta>0$ with $\Delta(p)>0$ for all $p\in(p^*-\delta,p^*]$, an interval of positive length.
\end{proof}

Because co-monotonicity passes to sub-vectors, sub-additivity of the whole also rules out diversification inside it.
\begin{corollary}
\label{coro:subportfolios}
Let $\XX$ satisfy Assumption~\ref{ass:zero-endpoint} and be VaR sub-additive. Then, for every non-empty $J\subseteq\{1,\dots,n\}$,
\[
\VaR_p\left[\sum_{i\in J}X_i\right]=\sum_{i\in J}\VaR_p[X_i],
\qquad \forall p\in(0,1).
\]
Conversely, if a single pair $(X_i,X_j)$, $i\neq j$, fails to be co-monotonic, then $\XX$ is not VaR sub-additive.
\end{corollary}
\begin{proof}
By Theorem~\ref{thm:VaRsubadditive} the vector $\XX$ is co-monotonic, so by Remark~\ref{rm:comonotonicity-facts}(a) there is a uniform random variable $U$ with $X_i=\VaR_U[X_i]$ almost surely for every $i$. Restricting this representation to the indices in $J$ shows that $(X_i)_{i\in J}$ is co-monotonic, and Lemma~\ref{lm:VaRtoolkit}(iii) delivers the stated additivity. The converse statement is the contrapositive of the same argument applied to $J=\{i,j\}$, combined with Remark~\ref{rm:comonotonicity-facts}(c).
\end{proof}

Thus, for lower-bounded risks, VaR sub-additivity at all levels means perfect lockstep, and no sub-portfolio benefits from diversification. By Corollary~\ref{coro:genericsuperadditivity}, every other portfolio is strictly super-additive \emph{somewhere}; Section~\ref{sec:VaR_super-additivity} asks when it is super-additive \emph{everywhere}.

%%%%%%%%%%%%%%%%%%%%%%%%%%%%%%%%%%%%%%%%%%%%%%%%%%%%%%%%%%%%%%%%%%%%%%%%%%%%%%%%%%%%%%%%%%%%%%%%%%%%%
\section{VaR super-additivity}
\label{sec:VaR_super-additivity}

We now turn to super-additivity at every level. The regime is non-empty: consider the counter-monotonic vector
\[
\XX=\left(X,\frac{1}{X}\right),
\]
whose joint CDF is the Fr\'echet lower bound
\[
F_{\XX}(x_1,x_2)=\max\{F_{X}(x_1)+F_{1/X}(x_2)-1,\,0\},
\]
where $X$ follows a Type II Pareto distribution with CDF
\[
F_{X}(x)=1-\left(\frac{\theta}{\theta+x}\right)^{\alpha}, \qquad x\ge 0,\; \alpha,\theta>0.
\]
For simplicity, take $\alpha=1/2$ and $\theta=1$. The margins are not identically distributed: the first is Pareto~(II) with shape $1/2$, whereas the second satisfies $F_{1/X}(y)=\bigl(y/(1+y)\bigr)^{1/2}$, a Beta-prime law with shape $1/2$. Their VaRs are
\[
\VaR_p[X]=\frac{p\,(2-p)}{(1-p)^2},\quad\VaR_p\left[\dfrac{1}{X}\right]=\frac{p^2}{1-p^2}.
\]
Because $x\mapsto x+1/x$ attains its minimum at $x=1$, the aggregate satisfies $S\ge 2$, and solving $x^2-sx+1\le 0$ gives $F_S(s)=F_X(r_+)-F_X(r_-)$ with $r_{\pm}=\bigl(s\pm\sqrt{s^2-4}\bigr)/2$. A direct calculation then yields
\[
\VaR_p\left[X+\dfrac{1}{X}\right]=\frac{2\left(1+2p^2-p^4\right)}{\left(1-p^2\right)^2},
\]
while the sum of the marginal VaRs is
\[
\VaR_p[X]+\VaR_p\!\left[\frac1X\right]=\frac{2 p\, \left(1+p-p^2\right)}{(1-p)^2 (1+p)}.
\]
Subtracting the two expressions, the numerators differ by exactly $2(1-p)$, so that the diversification penalty admits the closed form
\begin{equation}
\label{eq:opening-gap}
\VaR_p\left[X+\frac1X\right]-\VaR_p[X]-\VaR_p\!\left[\frac1X\right]
=\frac{2}{(1-p)(1+p)^2}\;>\;0,
\qquad \forall\, p\in(0,1),
\end{equation}
and \(\XX\) is VaR super-additive at every level; the penalty attains its minimum \(\tfrac{27}{16}\) at \(p=\tfrac13\) and diverges as \(p\uparrow1\). This vector belongs to the class developed below (see the discussion after Corollary~\ref{coro:NLOD_PhiNI}).

The recent literature on global super-additivity works with the stochastic dominance relation
\begin{equation}
\label{eq:stochdominance}
X\le_{st}\,\sum_{i=1}^n w_iX_i,
\end{equation}
where $X_1,\dots,X_n$ are identically distributed as $X$, independent or negatively dependent, with weights $w_i\in(0,1)$ and $\sum_{i=1}^n w_i=1$.

Here \citet{Ibragimov2009} treats one-sided stable laws, \citet{Chen2025} weakly negatively associated super-Pareto risks, \citet{Chen2026} NLOD risks for which \(\kappa(x)=-\log F_X(1/x)\) is sub-additive, and \citet{Arab2025} and \citet{Mueller2025} the InvSub and super-Cauchy classes under independence; see also \citet{Chen2025a,Zeng2025}. Relation \eqref{eq:stochdominance} presupposes a single parent law. For heterogeneous margins, \citet{Chen2026} and \citet{Vincent2025} compare the weighted average with the mixture \(\sum_iw_iF_{X_i}\) of the margins, for NLOD super-Fr\'echet and for independent completely subscalable risks respectively. Our aim is a criterion on the joint law that reaches dependence beyond NLOD, including positive dependence between some of the components. Integrability rules out genuine super-additivity:
\begin{remark}
\label{rm:super-additivityImamura}
Let $\XX$ be any random vector -- Assumption~\ref{ass:zero-endpoint} is not needed here -- with $\E[|X_i|]<\infty$ for all $i\in\{1,\dots,n\}$. If $\XX$ is VaR super-additive, then $\XX$ is VaR additive and co-monotonic. Indeed, by Lemma~\ref{lm:VaRtoolkit}(iv) applied with $c=0$,
\[
\VaR_p[-S]=-\VaRplus_{1-p}[S],
\qquad
\VaR_p[-X_i]=-\VaRplus_{1-p}[X_i],
\]
so that $-\XX$ is VaR sub-additive precisely when $\VaRplus_q[S]\ge\sum_{i=1}^n\VaRplus_q[X_i]$ for every $q\in(0,1)$; by Lemma~\ref{lm:VaRversion} this last statement is equivalent to VaR super-additivity of $\XX$. Since the components of $-\XX$ are integrable, Theorem~1 of \citet{Imamura2025} applies and yields that $-\XX$ is co-monotonic. Co-monotonicity is preserved under coordinatewise reflection, because the negative of a chain in $\RR^n$ is again a chain (Remark~\ref{rm:comonotonicity-facts}(b)), so $\XX$ is co-monotonic and, by Lemma~\ref{lm:VaRtoolkit}(iii), VaR additive.
\end{remark}

Hence a genuinely super-additive vector needs a non-integrable component, and in fact two (Theorem~\ref{thm:intrinsic-heavy}). The converse fails, since non-integrability constrains only the far tail:
\begin{example}
\label{example:VaRfailssuperadditivity}
Let \(\XX\) be a bivariate random vector.
\begin{itemize}
\item[(1)] Suppose \(\XX\) is counter-monotonic and defined by
\[
\XX = \left( X,\; \frac{1}{1+X} \right),
\]
where \(X\) follows a Pareto Type II distribution with unit scale and unit shape.
Since
\[
\frac{1}{1+X} \sim \mathrm{Unif}(0,1),
\]
we obtain \(\mathbb{E}\!\left[\dfrac{1}{ 1+X}\right] = \dfrac{1}{2}\), while \(\mathbb{E}[X] = \infty\) for this Pareto distribution.

The marginal VaR functions are therefore
\[
\VaR_p[X] = \frac{p}{1-p},
\qquad
\VaR_p\!\left[\frac{1}{1+X}\right] = p.
\]

The map \(x\mapsto x+1/(1+x)\) has derivative \(1-(1+x)^{-2}\ge0\) on \([0,\infty)\), hence is non-decreasing, so \(S\) is a non-decreasing transformation of \(X\) and its quantiles are obtained by composition:
\[
\VaR_p[S] = \frac{p}{1-p}+(1-p)=\frac{p^2 - p + 1}{1 - p}.
\]
Consequently the aggregation gap is the elementary expression
\begin{equation}
\label{eq:example-gap-1}
\VaR_p[S]-\VaR_p[X]-\VaR_p\!\left[\frac{1}{1+X}\right]=1-2p,
\end{equation}
so that VaR is super-additive for \(p\in\left(0,\tfrac12\right]\) and sub-additive for \(p\in\left[\tfrac12,1\right)\).

\item[(2)] Consider now a bivariate random vector \(\XX = (X_1, X_2)\) with joint distribution function
\[
F_{\XX}(x_1,x_2) = C\!\left(F_{X_1}(x_1),\, F_{X_2}(x_2)\right),
\]
where \(F_{X_1}\) and \(F_{X_2}\) are Pareto (II) marginal distributions with parameters 
\(\alpha_1=\alpha_2=\theta_1=\theta_2=1\). The copula \(C(u,v)\) is an {Ordinal Sum} copula \citep[Section~3.2]{10.5555/1952073}, given by
\[
C(u,v)=
\begin{cases}
\max\!\left\{u+v-\dfrac{1}{2},\,0\right\}, 
& (u,v)\in \left[0,\dfrac{1}{2}\right]^2,\\[0.6em]
\max\!\left\{u+v-1,\,\dfrac{1}{2}\right\}, 
& (u,v)\in \left(\dfrac{1}{2},1\right]^2,\\[0.6em]
\min\{u,v\}, 
& \text{otherwise}.
\end{cases}
\]
\begin{figure}[htb!]
\centering
\includegraphics[width=0.42\textwidth]{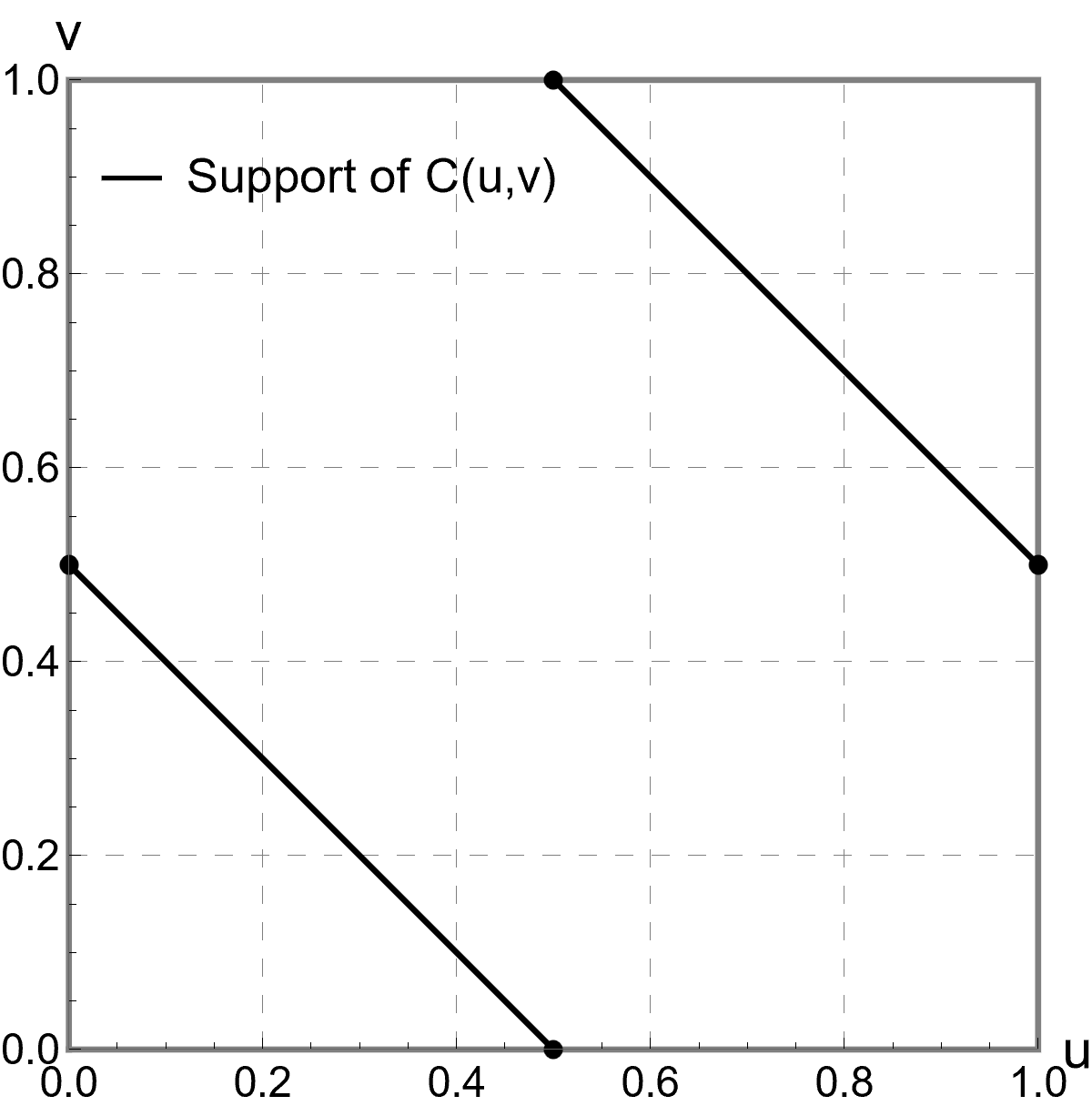}
\caption{Support of the Ordinal Sum copula \(C(u,v)\) on the unit square \([0,1]^2\).}
\label{fig:ordinal-sum-support}
\end{figure}

The copula places half of its mass counter-monotonically on each of the two anti-diagonal segments $\{u+v=\tfrac12\}$ and $\{u+v=\tfrac32\}$. Since both margins are Pareto (II) with unit shape, their expectations are infinite and their common VaR is
\[
\VaR_p[X_1]=\VaR_p[X_2]=\frac{p}{1-p}.
\]
On the lower segment the aggregate ranges over $\left[\tfrac23,1\right]$ and on the upper segment over $[6,\infty)$, so that
\[
F_S(s)=
\begin{cases}
\dfrac{\sqrt{9s^2+12s-12}}{2s+4}, & \dfrac23\le s\le 1,\\[0.8em]
\dfrac12, & 1<s<6,\\[0.8em]
\dfrac12+\dfrac{\sqrt{s^2-4s-12}}{2s+4}, & s\ge 6,
\end{cases}
\]
and inverting each branch gives
\[
\VaR_p[S] =
\begin{cases}
\dfrac{6 + 8p^2}{\,9 - 4p^2\,}, 
& 0 < p \le \dfrac{1}{2},\\[0.8em]
\dfrac{2 - 2p\,(1-p)}{p\,(1-p)}, 
& \dfrac{1}{2} < p < 1.
\end{cases}
\]

On the upper branch, $\VaR_p[S]-\VaR_p[X_1]-\VaR_p[X_2]=2/p>0$. On the lower branch, the comparison reduces to the quadratic inequality $4p^2-12p+3\le0$, whence \(\VaR_p[S]\) is {sub-additive} exactly for
\[
p \in \left[\frac{3-\sqrt{6}}{2},\,\frac{1}{2}\right]\approx[0.2753,\,0.5],
\]
and {super-additive} for all remaining values of \(p\).
\end{itemize}
\end{example}

The two parts fail for opposite reasons -- part~(1) through a margin, part~(2) through the dependence (Remark~\ref{rm:indispensable}) -- so super-additivity is a property of the interaction between dependence and margins.

\subsection{An exact criterion, and its factorization}
\label{subsec:criterion}
It is convenient to fix notation once and for all. For $p\in(0,1)$ write
\begin{equation}
\label{eq:quantilepath}
x_i(p) := \VaR_p[X_i],
\qquad
t(p) := \sum_{i=1}^n x_i(p),
\end{equation}
so that VaR super-additivity is the statement $\VaR_p[S]\ge t(p)$ for all $p\in(0,1)$. Let $U$ be uniform on $(0,1)$ and let $S^{c}:=\sum_{i=1}^n\VaR_U[X_i]$ be the aggregate of the co-monotonic vector with the same margins. By Lemma~\ref{lm:VaRtoolkit}(iii), $\VaR_p[S^c]=t(p)$, and first-order stochastic dominance is the pointwise order of quantile functions \citep[Section~1.A]{Shaked2007}; hence
\begin{equation}
\label{eq:comonotone-benchmark}
\XX\ \text{is VaR super-additive}\qquad\Longleftrightarrow\qquad S\ \ge_{st}\ S^{c}.
\end{equation}
The sum of stand-alone capital requirements is thus the capital of the co-monotonic portfolio, and global super-additivity asks the actual aggregate to be stochastically riskier than this benchmark; for identically distributed risks with weights \(w_i\), \eqref{eq:comonotone-benchmark} is \eqref{eq:stochdominance}. The following lemma converts this comparison of quantiles into a comparison of distribution functions, and is the engine of everything that follows. Its only input is the left-continuity of quantile functions recorded in Lemma~\ref{lm:VaRtoolkit}(i).
\begin{lemma}[Aggregation criterion]
\label{lm:aggregationcriterion}
Let $\XX$ be any random vector, let $t(\cdot)$ be as in \eqref{eq:quantilepath}, and write $F_S(t^-)=\lim_{z\uparrow t}F_S(z)$.
\begin{itemize}
\item[(i)] $\XX$ is VaR super-additive if and only if $F_S\bigl(t(p)^-\bigr)\le p$ for every $p\in(0,1)$.
\item[(ii)] In particular, if $F_S\bigl(t(p)\bigr)\le p$ for every $p\in(0,1)$, then $\XX$ is VaR super-additive; the converse holds when $F_S$ is continuous.
\end{itemize}
\end{lemma}
\begin{proof}
(i) Suppose $\XX$ is VaR super-additive and fix $p\in(0,1)$. Every $z<t(p)$ satisfies $z<\VaR_p[S]$, hence $F_S(z)<p$ by the definition of $\VaR_p[S]$; letting $z\uparrow t(p)$ gives $F_S(t(p)^-)\le p$. Conversely, suppose the left-limit condition holds, fix $p\in(0,1)$ and let $q\in(0,p)$. For every $z<t(q)$ we have $F_S(z)\le F_S(t(q)^-)\le q<p$. Since $\{z\in\RR : F_S(z)\ge p\}$ is an interval unbounded above with infimum $\VaR_p[S]$, any point at which $F_S$ is strictly below $p$ lies weakly to the left of $\VaR_p[S]$; hence $z\le\VaR_p[S]$ for every $z<t(q)$, that is, $t(q)\le\VaR_p[S]$. The map $q\mapsto t(q)$ is a finite sum of left-continuous functions and is therefore left-continuous, so letting $q\uparrow p$ gives $t(p)\le\VaR_p[S]$.

(ii) The first statement follows from (i) because $F_S(t^-)\le F_S(t)$, and the second because the two coincide when $F_S$ is continuous.
\end{proof}

We insert the independence coupling between \(F_S\) and \(p\), which splits the criterion into a dependence condition and a margin condition:
\begin{definition}
\label{def:NSD}
We say $\XX$ is negative simplex dependent (NSD) if 
\begin{equation}
\label{eq:NSD}
F_S(t)\le \prod_{i=1}^nF_{X_i}(t),\qquad \forall t\in[0,\infty).
\end{equation}
\end{definition}
\begin{definition}
\label{def:SD}
A function $\Phi: [0,\infty)^n\to [-\infty,0]$ is called simplex dominant (SD) if
\begin{equation}
\label{eq:SD}
\Phi(x_1,\dots,x_n)\ge \Phi(t,\dots,t),\qquad t=\sum_{i=1}^n x_i,\quad \forall (x_1,\dots,x_n)\in [0,\infty)^n.
\end{equation}
\end{definition}

The point \((x_1,\dots,x_n)\) lies on the simplex \(\{\xx\ge\oo:\sum_i x_i=t\}\), whence the name. NSD has a transparent probabilistic reading.
\begin{lemma}
\label{lm:NSDmax}
Let $X_1^{\perp},\dots,X_n^{\perp}$ be independent with $X_i^{\perp}\eqd X_i$, and set $M^{\perp}=\max\{X_1^{\perp},\dots,X_n^{\perp}\}$. Under Assumption~\ref{ass:zero-endpoint},
\[
\XX \text{ is NSD}
\qquad\Longleftrightarrow\qquad
S\;\ge_{st}\;M^{\perp}.
\]
In particular, NSD implies $\VaR_p[S]\ge\VaR_p[M^{\perp}]$ for every $p\in(0,1)$.
\end{lemma}
\begin{proof}
By independence, $F_{M^{\perp}}(t)=\prod_{i=1}^n F_{X_i}(t)$ for all $t\in\RR$, so \eqref{eq:NSD} states that $F_S\le F_{M^{\perp}}$ on $[0,\infty)$; both functions vanish on $(-\infty,0)$ because $S\ge0$ and $M^{\perp}\ge0$ almost surely, so the inequality holds on all of $\RR$, which is the definition of $S\ge_{st}M^{\perp}$. The quantile statement is the standard equivalent formulation of first-order stochastic dominance \citep[Chapter 1]{Shaked2007}.
\end{proof}

Thus NSD asks the dependent sum to dominate the maximum of independent copies -- a finite-level, one-directional version of the max--sum equivalence of subexponential laws \citep{Embrechts1980,Embrechts1997}. Unlike negative lower orthant dependence, it involves the margins as well as the copula. The next proposition gives a copula-only sufficient condition.

\begin{proposition}
\label{prop:dependenceNSD}
When $\XX$ is negative lower orthant dependent (NLOD) \citep{Block1982a,Joe1997}, that is
\begin{equation*}
F_{\XX}(x_1,\dots,x_n)\le \prod_{i=1}^n F_{X_i}(x_i),\ \forall (x_1,\dots,x_n)\in[0,\infty)^n,
\end{equation*}
then $\XX$ is NSD.
\end{proposition}
\begin{proof}
Under Assumption~\ref{ass:zero-endpoint} every component is non-negative, so on the event $\{S\le t\}$ we have $X_i\le\sum_{j=1}^n X_j\le t$ for each $i$; equivalently, the $n$-simplex $\{\xx\ge\oo:\sum_i x_i\le t\}$ is contained in the $n$-cube $[0,t]^n$. Hence
\begin{equation*}
F_S(t)\le F_{\XX}(t,\dots,t)\le \prod_{i=1}^n F_{X_i}(t),\qquad \forall t\in[0,\infty).
\end{equation*}
In fact, only the value of $F_{\XX}$ along the diagonal $(t,\dots,t)$, $t\in[0,\infty)$, is used, so it suffices that $\XX$ be NLOD on the diagonal rather than everywhere.
\end{proof}

The second half of the split is a condition on the margins alone. The next lemma extracts exactly what the simplex dominance of $\Phi$ contributes, and nothing more.
\begin{lemma}
\label{lm:SDimpliesQS}
Suppose each $F_{X_i}$ is continuous and the function
\begin{equation} 
\label{eq:Phiexpression}
\Phi(x_1,\dots,x_n)=\sum_{i=1}^n x_i \log F_{X_i}(x_i),
\end{equation} 
is SD. Then, with $t(\cdot)$ as in \eqref{eq:quantilepath},
\begin{equation}
\label{eq:QS}
\prod_{i=1}^n F_{X_i}\bigl(t(p)\bigr)\;\le\;p,
\qquad \forall p\in(0,1).
\end{equation}
\end{lemma}
\begin{proof}
Fix $p\in(0,1)$. By Remark~\ref{rm:positive-quantiles}, continuity of the margins together with Assumption~\ref{ass:zero-endpoint} gives $F_{X_i}(x_i(p))=p$ for each $i$ and $t(p)>0$; moreover $F_{X_i}(x)>0$ for every $x>0$, so all logarithms below are finite. Evaluating \eqref{eq:SD} at the point $(x_1(p),\dots,x_n(p))$, whose coordinates sum to $t(p)$, yields
\[
\sum_{i=1}^n x_i(p)\log F_{X_i}\bigl(x_i(p)\bigr)
\;\ge\;
t(p)\sum_{i=1}^n \log F_{X_i}\bigl(t(p)\bigr),
\]
that is, $t(p)\log p\ge t(p)\sum_{i=1}^n\log F_{X_i}(t(p))$. Dividing by $t(p)>0$ and exponentiating gives \eqref{eq:QS}.
\end{proof}

By Lemma~\ref{lm:NSDmax}, condition \eqref{eq:QS} reads $F_{M^{\perp}}(t(p))\le p$. The argument of Lemma~\ref{lm:aggregationcriterion}, with $M^{\perp}$ in place of $S$, shows that it implies $\VaR_p[M^{\perp}]\ge t(p)$ for every $p$, that is, $M^{\perp}\ge_{st}S^c$, and that the two are equivalent when the margins are continuous. The main result of the paper now follows by composing the two halves.
\begin{theorem}
\label{thm:Varsuperadditivity}
Let $\XX$ satisfy Assumption~\ref{ass:zero-endpoint}, suppose that $\XX$ is NSD, and suppose that the margins satisfy the quantile-simplex condition \eqref{eq:QS}. Then $\XX$ is VaR super-additive.

In particular, if $\XX$ is NSD with continuous $F_{X_i}$ and the function $\Phi$ of \eqref{eq:Phiexpression} is SD, then $\XX$ is VaR super-additive.
\end{theorem}
\begin{proof}
Fix $p\in(0,1)$. Evaluating the NSD inequality \eqref{eq:NSD} at $t=t(p)\ge0$ and then invoking \eqref{eq:QS} gives
\[
F_S\bigl(t(p)\bigr)
\;\le\;\prod_{i=1}^n F_{X_i}\bigl(t(p)\bigr)
\;\le\;p .
\]
Since $p\in(0,1)$ was arbitrary, Lemma~\ref{lm:aggregationcriterion}(ii) yields $\VaR_p[S]\ge t(p)=\sum_{i=1}^n\VaR_p[X_i]$ for every $p\in(0,1)$, which is VaR super-additivity. The second assertion follows because, under continuity of the margins, Lemma~\ref{lm:SDimpliesQS} shows that simplex dominance of $\Phi$ implies \eqref{eq:QS}.
\end{proof}

\begin{remark}
\label{rm:architecture}
In the language of \eqref{eq:comonotone-benchmark}, Theorem~\ref{thm:Varsuperadditivity} interposes the independent maximum between the aggregate and its co-monotonic benchmark:
\[
S\ \ge_{st}\ M^{\perp}\ \ge_{st}\ S^{c}.
\]
The first dominance is NSD, a property of the joint law; the second follows from \eqref{eq:QS}, a property of the margins alone. The price of the intermediate benchmark is that the two conditions are sufficient but not necessary (Example~\ref{example:notnecessary}). Nothing is lost when \(M^{\perp}\eqd S^c\), as for unit Fr\'echet margins, where NSD becomes necessary as well (Proposition~\ref{prop:frechet-boundary}). Moreover, \eqref{eq:QS} uses SD only along the quantile path \(p\mapsto(x_1(p),\dots,x_n(p))\), and neither NSD nor \eqref{eq:QS} requires continuous margins.
\end{remark}

The first link of the chain is quantitative: it bounds the diversification penalty of every NSD vector from below by a quantity that involves the margins only.
\begin{corollary}
\label{coro:penalty-bound}
If $\XX$ satisfies Assumption~\ref{ass:zero-endpoint} and is NSD, then, for every $p\in(0,1)$,
\[
\VaR_p[S]-\sum_{i=1}^n\VaR_p[X_i]\ \ge\ \VaR_p[M^{\perp}]-\sum_{i=1}^n\VaR_p[X_i].
\]
For Fr\'echet margins $F_{X_i}(x)=\exp\{-(x/\theta_i)^{-\alpha}\}$ with a common index $\alpha\in(0,1]$ and scales $\theta_i>0$, the right-hand side equals
\[
\Bigl[\Bigl(\sum_{i=1}^n\theta_i^{\alpha}\Bigr)^{1/\alpha}-\sum_{i=1}^n\theta_i\Bigr]\bigl(\log(1/p)\bigr)^{-1/\alpha}\ \ge\ 0 ,
\]
which vanishes identically when $\alpha=1$ and, when $\alpha<1$, is strictly positive and diverges as $p\uparrow1$.
\end{corollary}
\begin{proof}
The inequality is the quantile form of Lemma~\ref{lm:NSDmax}. For the Fr\'echet margins, $F_{M^{\perp}}(t)=\exp\{-t^{-\alpha}\sum_i\theta_i^{\alpha}\}$ and $\VaR_p[X_i]=\theta_i(\log(1/p))^{-1/\alpha}$, which gives the expression. Its sign follows from $(\sum_i\theta_i)^{\alpha}\le\sum_i\theta_i^{\alpha}$ for $\alpha\le1$, with strict inequality when $\alpha<1$ and $n\ge2$.
\end{proof}
\noindent For identical scales the bound reads \(\VaR_p[S]\ge n^{1/\alpha}\VaR_p[X_1]\): two NSD lines with Fr\'echet index \(\tfrac12\) require at least twice the sum of their stand-alone capitals. For independent risks the bound is asymptotically sharp as \(p\uparrow1\), because \(S\) and \(M^{\perp}\) then have equivalent tails by the max--sum equivalence of regularly varying laws \citep{Embrechts1980,Embrechts1997}.

\subsection{The marginal condition}
\label{subsec:margins}
Condition \eqref{eq:SD} couples all coordinates. It decouples when \(\Phi\) is monotone.
\begin{proposition}
\label{prop:Phimargins} 
If $\Phi$ is non-increasing in the sense that 
\[
\Phi(x_1,\dots,x_n)\ge \Phi(y_1,\dots,y_n)
\quad\text{whenever } x_i\le y_i \text{ for all } i\in\{1,\dots,n\},
\]
then $\Phi$ is SD. If moreover $\Phi$ is additively separable,
\[
\Phi(x_1,\dots,x_n)=\sum_{i=1}^n\phi_i(x_i)
\quad\text{with each }\phi_i:[0,\infty)\to[-\infty,0]\text{ finite on }(0,\infty),
\]
then $\Phi$ is non-increasing if and only if each $\phi_i$ is non-increasing. Consequently, if all $\phi_i$ are non-increasing then $\Phi$ is SD.
\end{proposition}

\begin{proof}
First part. Fix $(x_1,\dots,x_n)\in[0,\infty)^n$ and set 
\[
y_1=\dots = y_n = t := \sum_{i=1}^n x_i.
\]
Since all coordinates are non-negative, $x_i\le t$ for every $i$, and the non-increasing property gives
\[
\Phi(x_1,\dots,x_n)\ge \Phi(y_1,\dots,y_n)=\Phi(t,\dots,t),
\]
which is \eqref{eq:SD}; hence $\Phi$ is SD.

Second part. If every $\phi_i$ is non-increasing, then $\Phi$, being a sum of non-increasing functions of separate coordinates, is non-increasing. Conversely, suppose $\Phi$ is non-increasing. Fix $i\in\{1,\dots,n\}$, take $x_i\le y_i$, and set $x_j=y_j=z$ for some $z>0$ and all $j\ne i$. Then 
\[
\Phi(x_1,\dots,x_n)=\phi_i(x_i)+\sum_{j\neq i}\phi_j(z)\;\ge\;\Phi(y_1,\dots,y_n)=\phi_i(y_i)+\sum_{j\neq i}\phi_j(z),
\]
and the common finite term $\sum_{j\ne i}\phi_j(z)$ cancels, giving $\phi_i(x_i)\ge \phi_i(y_i)$; that is, each $\phi_i$ is non-increasing. The final assertion combines this with the first part.
\end{proof}

For the aggregator \eqref{eq:Phiexpression}, \(\Phi=\sum_i\phi_i\) with \(\phi_i(x)=x\log F_{X_i}(x)\), so it suffices that each \(\phi_i\) be non-increasing on \([0,\infty)\), a condition checked one margin at a time. It has a clean interpretation as an ageing property of the reciprocal risk.
\begin{proposition}
\label{prop:equivalenceNonIncreasing}
Let $X_i$ satisfy Assumption~\ref{ass:zero-endpoint} and write $G_i=\log\circ F_{X_i}$ and $\Lambda_i=-G_i$ for the cumulative reversed hazard function. Consider the statements:
\begin{itemize}
\item[(i)] $\phi_i(x_i)=x_i\log F_{X_i}(x_i)$ is non-increasing on $[0,\infty)$;
\item[(ii)] $x_i\mapsto F_{X_i}(x_i)^{x_i}$ is non-increasing on $[0,\infty)$;
\item[(iii)] $G_i$ satisfies the scale-shrinking property
\[
G_i(\lambda x_i)\le \dfrac{1}{\lambda}\,G_i(x_i),\ \quad \forall \lambda\in[1,\infty),\ \forall x_i\in[0,\infty);
\]
\item[(iv)] $1/X_i$ has decreasing failure rate on average (DFRA), that is, $y\mapsto H_{1/X_i}(y)/y$ is non-increasing, where $H_{1/X_i}=-\log\Fbar_{1/X_i}$;
\item[(v)] $F_{X_i}$ is differentiable and
\[
x_i\, r_{X_i}(x_i)\le \int_{x_i}^\infty r_{X_i}(w)\,\mathrm{d}w,\quad \forall x_i\in[0,\infty),
\]
where $r_{X_i}(x_i)=\dfrac{f_{X_i}(x_i)}{F_{X_i}(x_i)}$ is the reversed hazard rate function \citep{Block1998} of $X_i$.
\end{itemize}
Then (i), (ii) and (iii) are equivalent; (iv) is equivalent to them whenever $F_{X_i}$ is continuous, and (v) whenever $F_{X_i}$ is differentiable. Moreover, if $x_i^2\, r_{X_i}(x_i)$ is non-decreasing -- equivalently, if $1/X_i$ has a non-increasing hazard rate -- then (i)--(v) hold.
\end{proposition}
\begin{proof}
(i) $\Leftrightarrow$ (ii). Immediate from $\phi_i(x_i)=\log\bigl(F_{X_i}(x_i)^{x_i}\bigr)$ and the strict monotonicity of the logarithm.

(i) $\Leftrightarrow$ (iii). Let $0<x_i\le y_i$ and write $y_i=\lambda x_i$ with $\lambda\ge1$. Then
\begin{align*}
\phi_i(y_i)\le\phi_i(x_i)
\iff y_iG_i(y_i)&\le x_iG_i(x_i),
\\
\iff \lambda x_iG_i(\lambda x_i)&\le x_iG_i(x_i),
\\
\iff G_i(\lambda x_i)&\le \dfrac{1}{\lambda}G_i(x_i),
\end{align*}
the last step dividing by $x_i>0$. The case $x_i=0$ is trivial, since $\phi_i(0)=0\ge\phi_i(y_i)$ for every $y_i$.

(i) $\Leftrightarrow$ (iv). If $F_{X_i}$ is continuous then $\Fbar_{1/X_i}(y)=\P(X_i<1/y)=F_{X_i}(1/y)$ for $y>0$, so that $H_{1/X_i}(y)=\Lambda_i(1/y)$. Substituting $x_i=1/y$,
\[
\frac{H_{1/X_i}(y)}{y}=x_i\,\Lambda_i(x_i)=-\phi_i(x_i),
\]
and since $y\mapsto1/y$ reverses order, $H_{1/X_i}(y)/y$ is non-increasing in $y$ exactly when $-\phi_i$ is non-decreasing in $x_i$, that is, when $\phi_i$ is non-increasing.

(i) $\Leftrightarrow$ (v). If $F_{X_i}$ is differentiable then, using $\log F_{X_i}(x_i)=-\int_{x_i}^{\infty}r_{X_i}(w)\,\mathrm{d}w$,
\begin{align*}
\phi_i^{'}(x_i)&=\log F_{X_i}(x_i)+\dfrac{x_if_{X_i}(x_i)}{F_{X_i}(x_i)}
\\
&=-\int_{x_i}^\infty \dfrac{f_{X_i}(w)}{F_{X_i}(w)}\,\mathrm{d}w+\dfrac{x_if_{X_i}(x_i)}{F_{X_i}(x_i)}
\\
&=-\int_{x_i}^\infty r_{X_i}(w)\,\mathrm{d}w+x_i\,  r_{X_i}(x_i),
\end{align*}
so $\phi_i'\le0$ on $[0,\infty)$ is precisely the stated integral inequality.

Sufficiency of the monotonicity of $x_i^2\,r_{X_i}(x_i)$. If $x_i^2\, r_{X_i}(x_i)$ is non-decreasing then, for all $w\ge x_i>0$,
\[
w^2\, r_{X_i}(w)\ge x_i^2\, r_{X_i}(x_i)\quad\iff\quad r_{X_i}(w)\ge \dfrac{x_i^2\, r_{X_i}(x_i)}{w^2},
\]
and integrating over $w\in[x_i,\infty)$ gives
\[
\int_{x_i}^{\infty}r_{X_i}(w)\,\d w\ \ge\ x_i^2\, r_{X_i}(x_i)\int_{x_i}^{\infty}\dfrac{\d w}{w^2}\ =\ x_i\, r_{X_i}(x_i),
\]
which is (v). The parenthetical reformulation is established in Remark~\ref{rm:phi_r_X_sufficicent} below, and the implication is then the classical one that a decreasing failure rate entails a decreasing failure rate on average \citep[Chapter 2]{Lai2006}.
\end{proof}
\begin{remark}
\label{rm:phi_r_X_sufficicent}
Since \(h_{1/X_i}(x)=x^{-2}r_{X_i}(1/x)\), the condition that \(x^2r_{X_i}(x)\) be non-decreasing says exactly that \(1/X_i\) has a decreasing failure rate (DFR), whereas Proposition~\ref{prop:equivalenceNonIncreasing}(iv) shows that the weaker DFRA property is necessary and sufficient; the gap is the classical one between the DFR and DFRA classes \citep{Lai2006}. The DFR condition is convenient in practice: for the log-logistic law \(F_{X_i}(x)=x^{\alpha_i}/(1+x^{\alpha_i})\) one has \(x^2r_{X_i}(x)=\alpha_ix/(1+x^{\alpha_i})\), which is non-decreasing if and only if \(\alpha_i\le1\). Regularity of \(1/X_i\) also underlies the conditions of \citet{Chen2026} and \citet{Arab2025}, which concern \(-\log F_X(1/x)\) and \(\Fbar_X(1/x)\), respectively.
\end{remark}

\begin{remark}
\label{rm:starFrechet}
Let \(F_{X_i}\) be continuous and strictly increasing on \((0,\infty)\), and let \(Z\) be unit Fr\'echet, \(F_Z(z)=e^{-1/z}\). Then \(X_i\eqd g_i(Z)\) with \(g_i(z)=\VaR_{F_Z(z)}[X_i]\) continuous, strictly increasing and \(g_i(0^+)=0\), and \(\Lambda_i(g_i(z))=1/z\). Hence \(x\Lambda_i(x)=g_i(z)/z\) at \(x=g_i(z)\), so \(\phi_i\) is non-increasing if and only if \(z\mapsto g_i(z)/z\) is non-decreasing, that is, \(g_i\) is \emph{star-shaped at the origin}. The margins admitted by Proposition~\ref{prop:Phimargins} are therefore exactly the star-shaped transformations of the unit Fr\'echet law. This class strictly contains the super-Fr\'echet class of \citet{Chen2026}, which consists of convex transformations fixing the origin; for instance \(g(z)=z(1+\min\{z,1\})\) is star-shaped but not convex.
\end{remark}

\begin{example}
\label{example:Phi-non-increasing-1}
Table~\ref{tab:margins} lists continuous laws commonly used for claim severities, together with the exact parameter range on which \(\phi_i\) is non-increasing; derivations are given in Appendix~\ref{app:margins}. For the Fr\'echet, Pareto (II), inverse gamma and log-logistic families the admissible range is exactly the infinite-mean range, so Proposition~\ref{prop:phi_mean_infinite} below is sharp. The log-Cauchy law is admissible up to \(\alpha^\star\approx1.0568\), and the L\'evy and Beta-prime laws for all parameter values.
\end{example}
\begin{table}[htb!]
\centering
\small
\renewcommand{\arraystretch}{1.2}
\begin{tabular}{@{}llll@{}}
\toprule
 & Law & \(F_{X_i}(x)\), \(x\ge0\) & \(\phi_i\) non-increasing iff\\
\midrule
(1) & Fr\'echet & \(\exp\{-(x/\theta_i)^{-\alpha_i}\}\) & \(0<\alpha_i\le1\)\\
(2) & Pareto (II)/Lomax & \(1-\bigl(\theta_i/(\theta_i+x)\bigr)^{\alpha_i}\) & \(0<\alpha_i\le1\)\\
(3) & L\'evy & \(\mathrm{erfc}\bigl(\sqrt{\theta_i/(2x)}\bigr)\) & all \(\theta_i>0\)\\
(4) & Beta prime \((\alpha_i,1)\) & \(\bigl(x/(1+x)\bigr)^{\alpha_i}\) & all \(\alpha_i>0\)\\
(5) & Log-power & \(\exp\{-[\log(1+x)]^{\alpha_i}/x\}\), \(\alpha_i<1\) & \(0\le\alpha_i<1\)\\
(6) & Log-Cauchy & \(\tfrac12+\tfrac1\pi\arctan(\alpha_i\log x)\) & \(0<\alpha_i\le\alpha^\star\approx1.0568\)\\
(7) & Inverse gamma & \(\Gamma(\alpha_i,\theta_i/x)/\Gamma(\alpha_i)\) & \(0<\alpha_i\le1\)\\
(8) & Log-logistic & \(x^{\alpha_i}/(1+x^{\alpha_i})\) & \(0<\alpha_i\le1\)\\
\bottomrule
\end{tabular}
\caption{Severity laws for which \(\phi_i(x)=x\log F_{X_i}(x)\) is non-increasing, so that \(\Phi\) in \eqref{eq:Phiexpression} is SD by Proposition~\ref{prop:Phimargins}. Here \(\alpha^\star=\inf_{u\in(0,1)}\{-\pi u\log u/\sin^2(\pi u)\}\), and \(\Gamma(\cdot,\cdot)\) is the upper incomplete gamma function.}
\label{tab:margins}
\end{table}

All laws in Table~\ref{tab:margins} have infinite means, and necessarily so; in fact their tails are at least as heavy as a Pareto tail of index one.
\begin{proposition}
\label{prop:phi_mean_infinite}
Let $X_i$ be non-degenerate and satisfy Assumption~\ref{ass:zero-endpoint}, and write $\psi_i:=-\phi_i$, so that $\psi_i(x)=x\,\Lambda_i(x)\ge0$. If $\phi_i$ is non-increasing on $[0,\infty)$, then
\[
\lim_{t\to\infty}t\,\Fbar_{X_i}(t)=\sup_{s>0}\psi_i(s)\in(0,\infty].
\]
In particular, $\Fbar_{X_i}(t)\ge c/t$ for some $c>0$ and all large $t$, and $\E[X_i]=\infty$.
\end{proposition}
\begin{proof}
The function $\psi_i$ is non-decreasing, so $\psi_i(t)\to\sup_s\psi_i(s)$ as $t\to\infty$. Since $X_i$ is non-degenerate with lower endpoint $0$, there is $x_0>0$ with $0<F_{X_i}(x_0)<1$, whence $\psi_i(x_0)>0$. The variable $X_i$ is unbounded: if $F_{X_i}(x_1)=1$ for some $x_1$, then $\psi_i(t)=0$ for $t\ge x_1$, which is impossible for a non-decreasing function with $\psi_i(x_0)>0$. Thus $\Fbar_{X_i}(t)\in(0,1)$ for all large $t$ and $\Fbar_{X_i}(t)\to0$, so that $\Lambda_i(t)=-\log\bigl(1-\Fbar_{X_i}(t)\bigr)\sim\Fbar_{X_i}(t)$ and $t\,\Fbar_{X_i}(t)/\psi_i(t)\to1$. This gives the limit, and any $c\in\bigl(0,\sup_s\psi_i(s)\bigr)$ yields $\Fbar_{X_i}(t)\ge c/t$ for large $t$, so that $\E[X_i]=\int_0^\infty\Fbar_{X_i}(t)\,\d t=\infty$.
\end{proof}

Monotonicity of each \(\phi_i\) is sufficient but not necessary for SD:
\begin{example}
\label{example:SDnotNI}
Let \(X_1,X_2,X_3\) be independent (hence NSD), with \(F_{X_1}(x)=F_{X_2}(x)=\exp(-x^{-1/2})\) (Fr\'echet, shape \(\tfrac12\)) and
\[
F_{X_3}(x)=\begin{cases} x^2/e, & 0\le x\le1,\\ \exp(-x^{-1/2}), & x>1.\end{cases}
\]
Then \(\phi_1(x)=\phi_2(x)=-\sqrt x\), whereas \(\phi_3(x)=x(2\log x-1)\) on \([0,1]\) and \(\phi_3(x)=-\sqrt x\) on \((1,\infty)\). Thus \(\phi_3\) increases on \([e^{-1/2},1]\) (Figure~\ref{fig:phi-nonSD}), and \(\Phi\) is not non-increasing. Nevertheless \(\Phi\) is SD; the verification, a three-case estimate, is given in Appendix~\ref{app:SDnotNI}. Since all margins are continuous, \(\XX\) is VaR super-additive by Theorem~\ref{thm:Varsuperadditivity}. Figure~\ref{fig:VaRsnonSD} confirms this numerically, using \(\VaR_p[X_1]=\VaR_p[X_2]=\log^{-2}(1/p)\), and \(\VaR_p[X_3]=\sqrt{ep}\) for \(p\le e^{-1}\) and \(\VaR_p[X_3]=\log^{-2}(1/p)\) otherwise.
\begin{figure}[htb!]
\centering
\begin{subfigure}{0.48\textwidth}
\centering
\includegraphics[width=\linewidth]{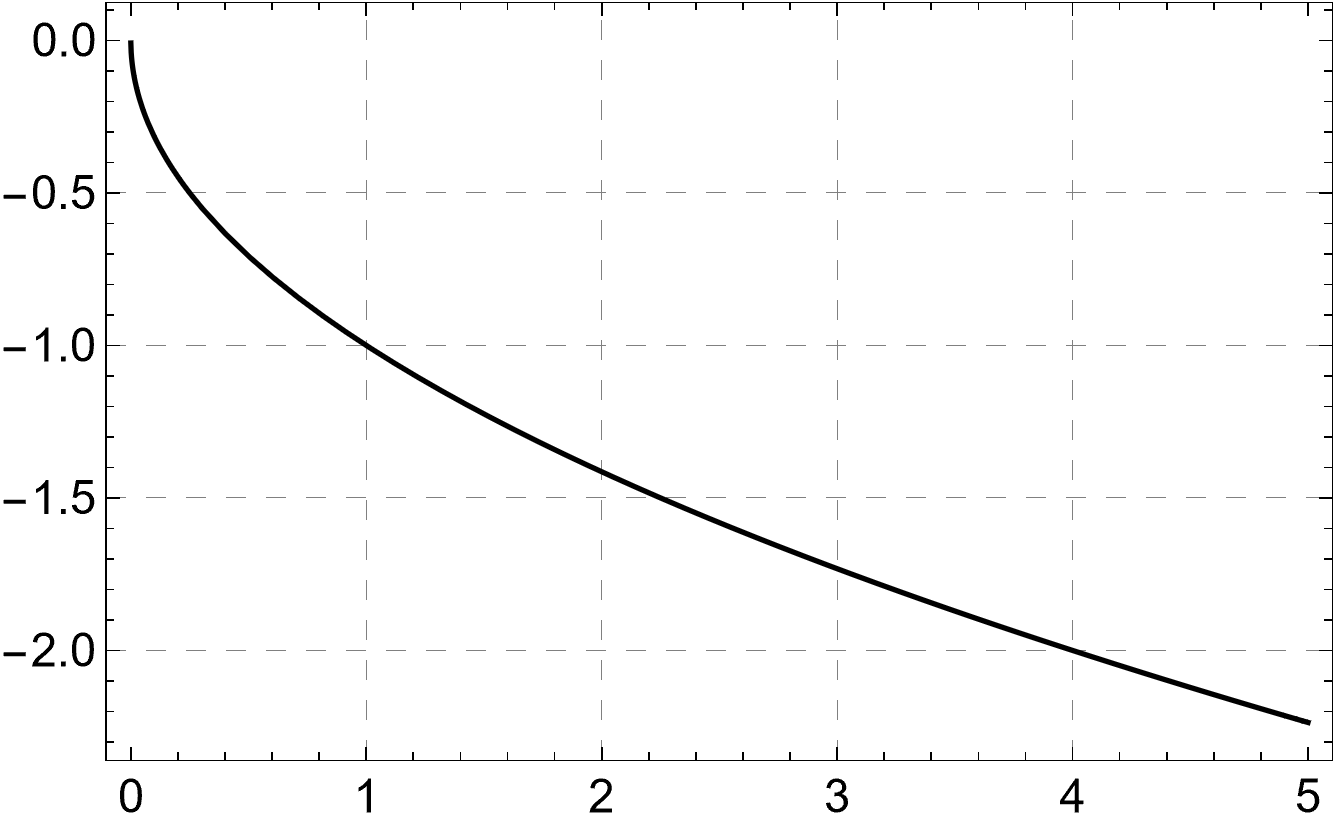}
\caption{$\phi_1=\phi_2$}
\label{subfig:phi12nonSD}
\end{subfigure}\hfill
\begin{subfigure}{0.48\textwidth}
\centering
\includegraphics[width=\linewidth]{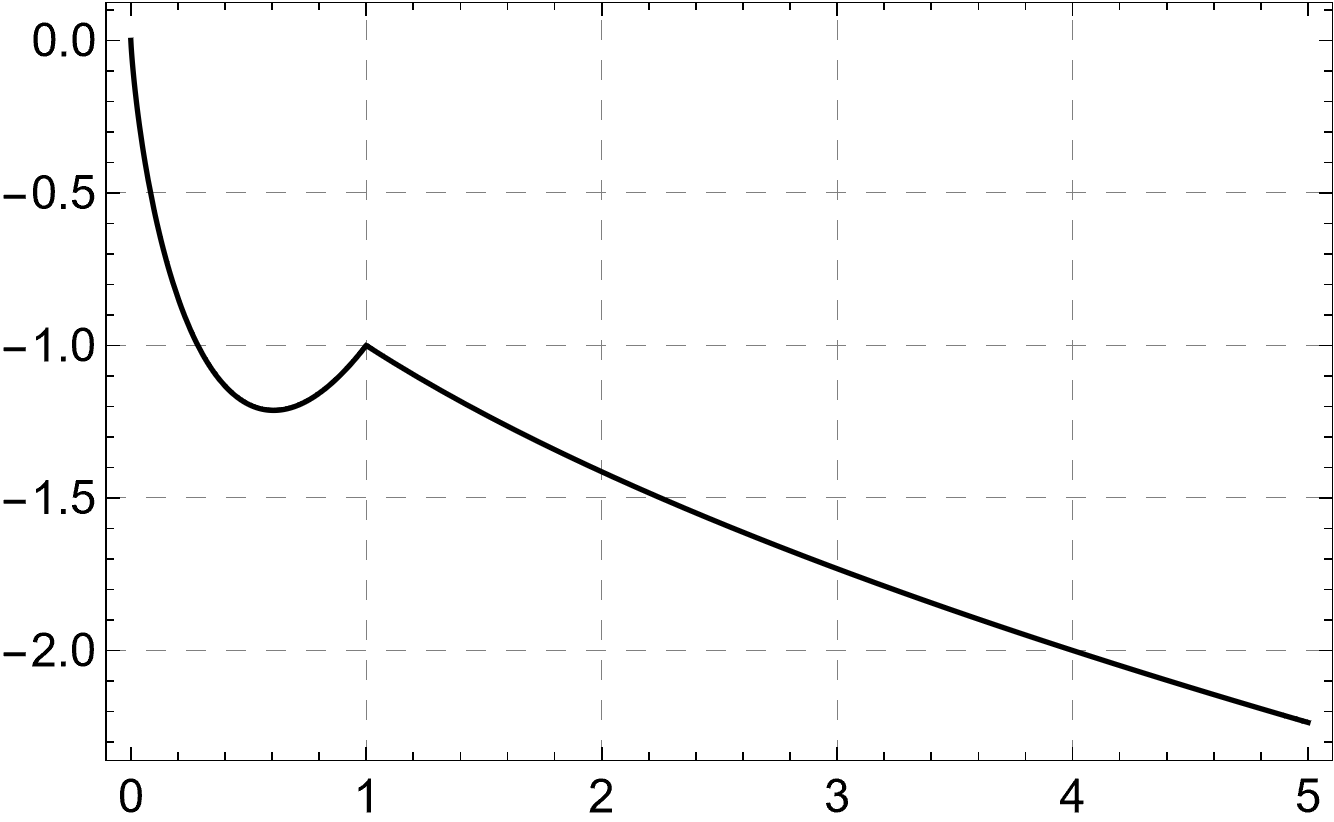}
\caption{$\phi_3$}
\label{subfig:phi3nonSD}
\end{subfigure}
\caption{The functions $\phi_i$ of Example~\ref{example:SDnotNI}.}
\label{fig:phi-nonSD}
\end{figure}
\begin{figure}[htb!]
\centering
\includegraphics[width=0.55\textwidth]{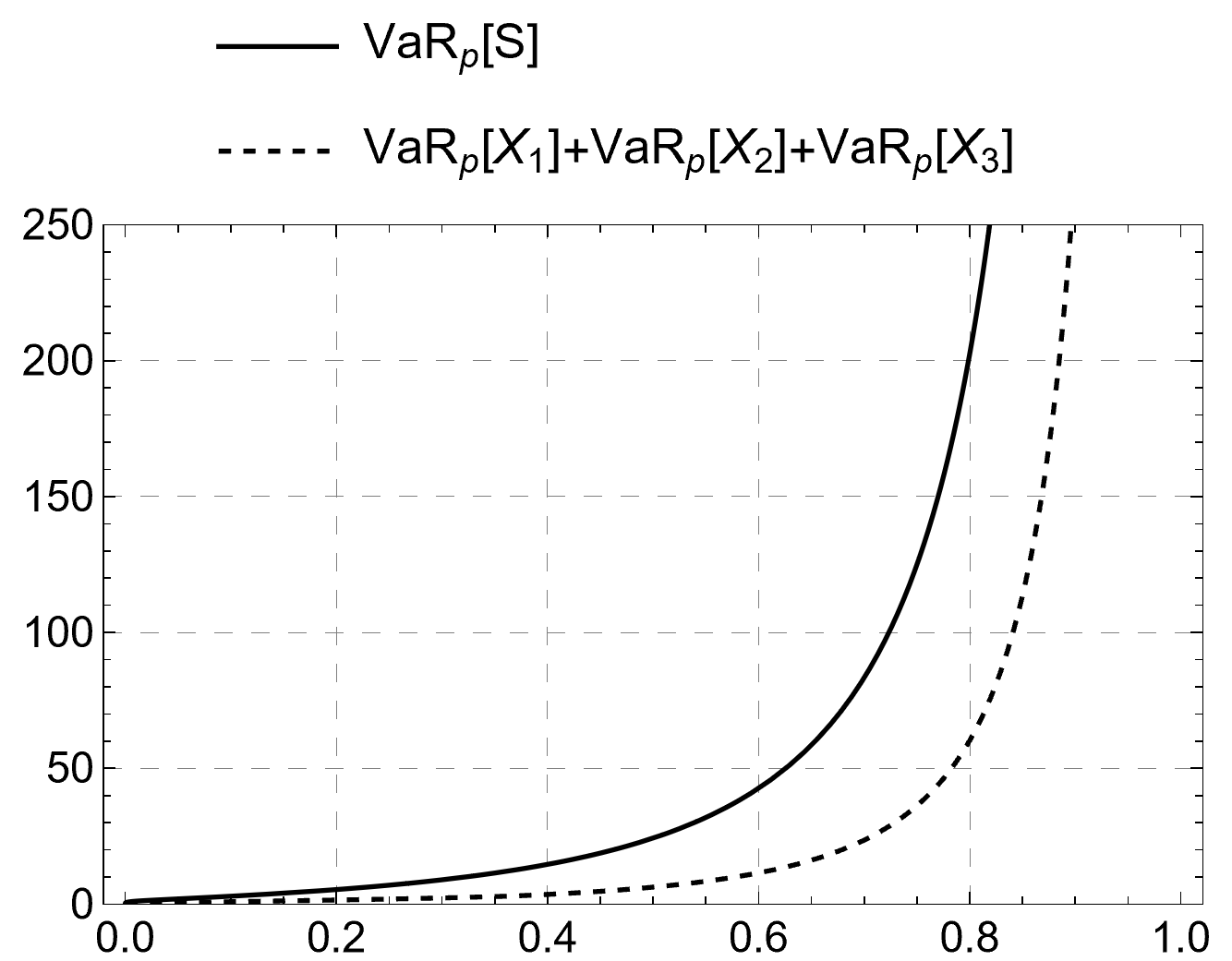}
\caption{$\VaR_p[S]$ versus $\sum_{i=1}^3\VaR_p[X_i]$ in Example~\ref{example:SDnotNI}.}
\label{fig:VaRsnonSD}
\end{figure}
\end{example}

\subsection{A separable criterion, and the reach of NSD}
\label{subsec:separable}
Combining Propositions~\ref{prop:dependenceNSD} and~\ref{prop:Phimargins} yields a criterion in which the copula and each margin are checked separately.
\begin{corollary}
\label{coro:NLOD_PhiNI}
If $\XX$ is NLOD with continuous $F_{X_i}$, and each $\phi_i(x_i)=x_i\log F_{X_i}(x_i)$ in Equation \eqref{eq:Phiexpression} is non-increasing, then $\XX$ is VaR super-additive.
\end{corollary}

Here NLOD concerns the copula alone and the monotonicity of \(\phi_i\) the \(i\)-th margin alone. The opening vector of this section is counter-monotonic, hence NLOD (as \(\max\{u+v-1,0\}\le uv\)), and its Pareto (II) and Beta-prime margins of shape \(\tfrac12\) appear in Table~\ref{tab:margins}; its super-additivity therefore follows from Corollary~\ref{coro:NLOD_PhiNI} without computation.

\begin{remark}
\label{rm:ChenShneer}
Relation \eqref{eq:stochdominance} is the special case of a common parent law. Let \(Y_i=w_iX_i\) with \(w_i\in(0,1)\), \(\sum_iw_i=1\), and \(X_1,\dots,X_n\) identically distributed and NLOD. Then \(\YY\) is NLOD and, since \(\sum_i\VaR_p[Y_i]=\VaR_p[X]\), its super-additivity is exactly \eqref{eq:stochdominance}. Along the ray \(y_i=w_it\), condition \eqref{eq:SD} for \(\YY\) reads
\begin{equation}
\label{eq:SD_theta_chen}
F_X(t)\ge \prod_{i=1}^n F_{X}\left(\frac{t}{w_i}\right),\quad\forall t\in[0,\infty),
\end{equation}
where \(X\eqd X_i\). With \(x_i=w_i/t\), requiring \eqref{eq:SD_theta_chen} for every \(\ww\) is exactly the sub-additivity \(\kappa(\sum_ix_i)\le\sum_i\kappa(x_i)\) of \(\kappa(x)=-\log F_X(1/x)\); together with NLOD, this is the class of \citet{Chen2026} for identically distributed risks. For continuous \(F_X\), \eqref{eq:SD_theta_chen} at \(t=\VaR_p[X]\) is precisely condition \eqref{eq:QS} for \(\YY\), so Theorem~\ref{thm:Varsuperadditivity} recovers their super-additivity result through \eqref{eq:QS} alone. Their result for non-identical NLOD super-Fr\'echet margins is recovered by Corollary~\ref{coro:NLOD_PhiNI} and Remark~\ref{rm:starFrechet}, since convex transformations fixing the origin are star-shaped; in exchange, \citet{Chen2026} obtain the stronger comparison with generalized means of the margins. Theorem~\ref{thm:Varsuperadditivity} goes further by allowing NSD instead of NLOD, and the one-dimensional condition \eqref{eq:QS} instead of \eqref{eq:SD}.
\end{remark}

Theorem~\ref{thm:Varsuperadditivity} reaches beyond Corollary~\ref{coro:NLOD_PhiNI}, because NSD is strictly weaker than NLOD:
\begin{example}
\label{example:NSDnotNLOD}
Let \(\XX=(X,X,1/X)\), where \(X\) is Pareto (II) with unit scale and shape, so that all three margins equal \(F(x)=x/(1+x)\). Since \(2x+1/x\ge2\sqrt2\), solving \(2x^2-sx+1\le0\) gives
\[
F_S(s)=\frac{\sqrt{s^2-8}}{s+3}\quad(s>2\sqrt2),\qquad F_S(s)=0\quad(s\le2\sqrt2).
\]
\emph{NSD holds.} For \(t>2\sqrt2\), squaring shows that \(F_S(t)\le F(t)^3\) is equivalent to \((1+t)^6(t^2-8)\le t^6(t+3)^2\), and indeed
\[
t^6(t+3)^2-(1+t)^6(t^2-8)=2t^6+28t^5+105t^4+154t^3+119t^2+48t+8>0 .
\]
\emph{NLOD fails.} For \(t>1\), \(F_{\XX}(t,t,t)=\P(1/t\le X\le t)=(t-1)/(t+1)\), so
\[
F_{\XX}(t,t,t)-F(t)^3=\frac{t^2-t-1}{(1+t)^3}\ge0\qquad\text{for } t\ge\tfrac{1+\sqrt5}{2},
\]
even on the diagonal. \emph{Super-additivity.} Inverting \(F_S\),
\[
\VaR_p[S]-3\VaR_p[X]=\frac{3p^2+\sqrt{p^2+8}}{1-p^2}-\frac{3p}{1-p}=\frac{\sqrt{p^2+8}-3p}{1-p^2}>0,\qquad p\in(0,1),
\]
as predicted by Theorem~\ref{thm:Varsuperadditivity}, since unit-shape Pareto (II) margins satisfy SD (Table~\ref{tab:margins}); Corollary~\ref{coro:NLOD_PhiNI} does not apply.
\end{example}

The factorization loses nothing at the boundary of Table~\ref{tab:margins}, the unit Fr\'echet law, for which \(\phi_i\) is constant. There the two benchmarks of Remark~\ref{rm:architecture} coincide, and super-additivity becomes a property of the dependence alone.
\begin{proposition}[The unit Fr\'echet boundary]
\label{prop:frechet-boundary}
Let $F_{X_i}(x)=\exp(-\theta_i/x)$, $x>0$, with $\theta_i>0$ for $i\in\{1,\dots,n\}$, and put $\theta_+=\sum_i\theta_i$. Then $M^{\perp}\eqd S^c$, both with distribution function $\exp(-\theta_+/t)$, $t>0$. Consequently, whatever the dependence, $\XX$ is VaR super-additive if and only if it is NSD, that is, if and only if
\[
\P(S\le t)\ \le\ \exp(-\theta_+/t),\qquad t>0 .
\]
\end{proposition}
\begin{proof}
By independence, $F_{M^{\perp}}(t)=\prod_i\exp(-\theta_i/t)=\exp(-\theta_+/t)$. Since $\VaR_p[X_i]=\theta_i/\log(1/p)$, we have $\VaR_p[S^c]=t(p)=\theta_+/\log(1/p)$, which is the quantile function of the same law. By \eqref{eq:comonotone-benchmark}, super-additivity is $S\ge_{st}S^c$, that is, $S\ge_{st}M^{\perp}$, which is NSD by Lemma~\ref{lm:NSDmax}.
\end{proof}

Proposition~\ref{prop:frechet-boundary} rests on the max-stability of the Fr\'echet law \citep{Embrechts1997}: at unit index, the maximum of independent risks and their co-monotonic sum are indistinguishable. Index one is also where the tail-asymptotic analysis of \citet{Embrechts2009b} places the transition between asymptotic VaR sub- and super-additivity. The proposition turns every NSD check for unit Fr\'echet margins into a characterization. The next model, with a catastrophe shock common to all lines, shows that positively and tail-dependent risks can be super-additive at every level.
\begin{example}[Common shocks]
\label{example:commonshock}
Let $Z_0,Z_1,\dots,Z_n$ be independent unit Fr\'echet variables, $\P(Z_j\le z)=e^{-1/z}$, and let
\[
X_i=\max\{\rho_iZ_0,\ \sigma_iZ_i\},\qquad \rho_i,\sigma_i\ge0,\quad \theta_i:=\rho_i+\sigma_i>0 .
\]
The variable $Z_0$ is a shock common to all lines, such as a natural catastrophe, and $Z_i$ an idiosyncratic loss. The model is the max-linear counterpart of the common-shock models of \citet{Marshall1967} and \citet{Lindskog2003}, and a discrete instance of the spectral representation of max-stable vectors \citep{deHaan1984}. Each margin is unit Fr\'echet, $F_{X_i}(x)=\P(Z_0\le x/\rho_i)\,\P(Z_i\le x/\sigma_i)=\exp(-\theta_i/x)$, and
\[
F_{\XX}(x_1,\dots,x_n)=\exp\Bigl(-\max_{i}\frac{\rho_i}{x_i}-\sum_{i=1}^n\frac{\sigma_i}{x_i}\Bigr)\ \ge\ \prod_{i=1}^nF_{X_i}(x_i),
\]
with strict inequality as soon as two of the $\rho_i$ are positive: the vector is positively, not negatively, lower orthant dependent. For $n=2$, since $\P(X_1>\theta_1s,\,X_2>\theta_2s)=\P\bigl(Z_0>s\max_i\theta_i/\rho_i\bigr)+O(s^{-2})$ as $s\to\infty$, its upper tail dependence coefficient \citep{Joe1997} equals $\min\{\rho_1/\theta_1,\rho_2/\theta_2\}$, which is positive when $\rho_1\rho_2>0$. Nevertheless it is NSD. Indeed $S\ge \rho_+Z_0$ and $S\ge\sum_i\sigma_iZ_i$, where $\rho_+=\sum_i\rho_i$, and $Z_0$ is independent of $(Z_1,\dots,Z_n)$, so that
\[
\P(S\le t)\ \le\ \P(\rho_+Z_0\le t)\,\P\bigl(\sigma_iZ_i\le t\ \ \forall i\bigr)\ =\ e^{-\rho_+/t}\prod_{i=1}^ne^{-\sigma_i/t}\ =\ e^{-\theta_+/t}.
\]
By Proposition~\ref{prop:frechet-boundary}, $\XX$ is VaR super-additive at every level. When $\sigma_i=0$ for all $i$ the vector is co-monotonic and the bound holds with equality, in line with Theorem~\ref{thm:VaRsubadditive}: a common shock can erode the diversification penalty, but it cannot turn it into a benefit.
\end{example}

\subsection{Closure properties and the sharpness of the conditions}
\label{subsec:closure}
We next ask which coordinatewise transformations preserve super-additivity.
As in Remark~\ref{rm:starFrechet}, call \(\xi_i:[0,\infty)\to[0,\infty)\) \emph{star-shaped at the origin} if \(\xi_i(0)=0\) and \(x\mapsto\xi_i(x)/x\) is non-decreasing on \((0,\infty)\). Every convex \(\xi_i\) with \(\xi_i(0)=0\) is star-shaped, but not conversely: \(\xi_i(x)=x\,g_i(x)\) is star-shaped for any positive non-decreasing \(g_i\) \citep{Lai2006}.
\begin{proposition}
\label{prop:VaR_X_convex}
Let $\XX$ be NLOD with continuous marginal distributions $F_{X_i}$, and suppose each $\phi_i$ in Equation \eqref{eq:Phiexpression} is non-increasing. Define
\[
\widetilde{\XX}=(\widetilde{X}_1,\dots,\widetilde{X}_n), \quad \text{where } \widetilde{X}_i = \xi_i(X_i).
\]
If each $\xi_i$ is strictly increasing and star-shaped at the origin, then $\widetilde{\XX}$ is VaR super-additive. In particular this holds when each $\xi_i$ is strictly increasing, convex, and satisfies $\xi_i(0) = 0$.
\end{proposition}

\begin{proof}
A strictly increasing star-shaped map need not be continuous: it may jump upwards, as \(\xi(x)=x\mathbbm{1}_{\{x<1\}}+2x\mathbbm{1}_{\{x\ge1\}}\) shows. We therefore work with the generalized inverse
\[
\xi_i^{\leftarrow}(y)=\sup\{u\ge0:\ \xi_i(u)\le y\},\qquad y\ge0 .
\]
Star-shapedness gives \(\xi_i(x)\le x\,\xi_i(1)\) for \(x\le1\) and \(\xi_i(x)\ge x\,\xi_i(1)\) for \(x\ge1\), so \(\xi_i(0^+)=0\) and \(\xi_i(x)\to\infty\). Consequently \(\xi_i^{\leftarrow}\) is finite, non-decreasing and continuous (a jump of \(\xi_i^{\leftarrow}\) would require a flat stretch of \(\xi_i\)), and \(\xi_i^{\leftarrow}(y)>0\) for \(y>0\). Since \(\xi_i\) is strictly increasing, \(\{u:\xi_i(u)\le y\}\) is an interval with right endpoint \(\xi_i^{\leftarrow}(y)\), and since \(X_i\) has no atoms,
\[
F_{\widetilde X_i}(y)=\P\bigl(\xi_i(X_i)\le y\bigr)=F_{X_i}\bigl(\xi_i^{\leftarrow}(y)\bigr),\qquad y\ge0 .
\]
Thus \(F_{\widetilde X_i}\) is continuous, vanishes at the origin and is positive on \((0,\infty)\), so \(\widetilde\XX\) satisfies Assumption~\ref{ass:zero-endpoint}. The same argument applied to orthants gives, by NLOD of \(\XX\),
\[
F_{\widetilde{\XX}}(y_1,\dots,y_n) = F_{\XX}\big(\xi_1^{\leftarrow}(y_1),\dots,\xi_n^{\leftarrow}(y_n)\big)
\le \prod_{i=1}^n F_{X_i}\big(\xi_i^{\leftarrow}(y_i)\big)
= \prod_{i=1}^n F_{\widetilde{X}_i}(y_i),
\]
so \(\widetilde{\XX}\) is NLOD.

Next define $\widetilde{\phi}_i(x) = x \log F_{\widetilde{X}_i}(x)$ and fix $0\le x<y$. If $x=0$ then $\widetilde\phi_i(x)=0\ge\widetilde\phi_i(y)$, so assume $x>0$ and set $u=\xi_i^{\leftarrow}(x)$ and $v=\xi_i^{\leftarrow}(y)$, so that $0<u\le v$. If $u=v$, then $F_{\widetilde X_i}(x)=F_{\widetilde X_i}(y)\le1$ and $\widetilde\phi_i(y)\le\widetilde\phi_i(x)$ because $y>x$. If $u<v$, the definition of the generalized inverse gives $x\le\xi_i(u^+)$ and $\xi_i(v^-)\le y$, and star-shapedness, applied at points $u<s<s'<v$ and passed to the limits $s\downarrow u$, $s'\uparrow v$, gives $\xi_i(u^+)/u\le\xi_i(v^-)/v$. Hence
\[
\frac{x}{u}\ \le\ \frac{\xi_i(u^+)}{u}\ \le\ \frac{\xi_i(v^-)}{v}\ \le\ \frac{y}{v},
\qquad\text{that is,}\qquad x\le y\,\frac{u}{v}.
\]
The non-increasing property of $\phi_i$ applied to $u<v$ gives $\log F_{X_i}(v)\le\frac{u}{v}\log F_{X_i}(u)$. Since $\log F_{X_i}(u)\le0$,
\[
\widetilde{\phi}_i(y)=y\log F_{X_i}(v)\ \le\ y\,\frac{u}{v}\,\log F_{X_i}(u)\ \le\ x\log F_{X_i}(u)=\widetilde{\phi}_i(x),
\]
so $\widetilde{\phi}_i$ is non-increasing on $[0,\infty)$.

Applying Corollary~\ref{coro:NLOD_PhiNI} to $\widetilde{\XX}$, we conclude that it is VaR super-additive. The final assertion follows because convexity together with $\xi_i(0)=0$ implies $\xi_i(\lambda x)\le\lambda\xi_i(x)$ for $\lambda\in[0,1]$, which is star-shapedness.
\end{proof}

Hence the class is stable under proportional loadings, power transforms \(x\mapsto x^{\beta}\) with \(\beta\ge1\), and other strictly increasing indexation maps fixing the origin whose average rate of increase does not decrease. In the language of Remark~\ref{rm:starFrechet}, the margin part of the proof is the fact that compositions of star-shaped maps are star-shaped.

Neither condition of Theorem~\ref{thm:Varsuperadditivity} can be dropped:
\begin{remark}
\label{rm:indispensable}
Example~\ref{example:VaRfailssuperadditivity} shows that neither hypothesis of Theorem~\ref{thm:Varsuperadditivity} may be discarded.

In part~(1) the vector is counter-monotonic, hence NLOD, hence NSD by Proposition~\ref{prop:dependenceNSD}: the dependence hypothesis holds in full. What fails is the margin condition. The second margin is uniform on $(0,1)$, so $\phi_2(x)=x\log\min\{x,1\}$ is strictly increasing on $\left[e^{-1},1\right]$, and SD itself breaks down, since with $x_1=x_2=\tfrac12$ and $t=1$,
\[
\Phi\left(\tfrac12,\tfrac12\right)=\tfrac12\log\tfrac13+\tfrac12\log\tfrac12\approx-0.896
\;<\;
-\log 2\approx-0.693=\Phi(1,1).
\]
In part~(2), both margins are Pareto (II) with unit shape, so \(\Phi\) is SD (Table~\ref{tab:margins}); as super-additivity fails, Theorem~\ref{thm:Varsuperadditivity} forces NSD to fail. Indeed \(F_S\equiv\tfrac12\) on \((1,6)\), while \(F_{X_1}(t)F_{X_2}(t)=\bigl(t/(1+t)\bigr)^2<\tfrac12\) for \(t<1+\sqrt2\); so \eqref{eq:NSD} fails on \((1,1+\sqrt2)\), and in fact on \((t^\star,1+\sqrt2)\) with \(t^\star\approx0.7013\).
\end{remark}

Nor is either condition necessary:
\begin{example}
\label{example:notnecessary}
~
\begin{itemize}
\item[(1)] \emph{SD without NSD.} Let \(\Fbar_{\XX}(x_1,x_2)=(1+x_1+x_2)^{-\alpha}\), \(x_1,x_2\ge0\), with \(0<\alpha<1\). Both margins are Pareto (II) with \(F_{X_i}(x)=1-(1+x)^{-\alpha}\), so \(\Phi\) is SD (Table~\ref{tab:margins}). Integrating the density \(\alpha(\alpha+1)(1+x_1+x_2)^{-\alpha-2}\) over the simplex gives \(F_S(s)=1-(1+s)^{-\alpha}\bigl(1+\alpha s/(1+s)\bigr)\). The mean value theorem for \(u\mapsto u^{\alpha}\) on \([1/(1+t),1]\) gives \(\alpha t/(1+t)<F_{X_1}(t)\) for \(t>0\), whence
\[
F_S(t)>1-(1+t)^{-\alpha}\bigl(2-(1+t)^{-\alpha}\bigr)=F_{X_1}(t)F_{X_2}(t),\qquad t>0,
\]
so NSD fails at every \(t>0\). Yet \(\XX\) is VaR super-additive. Since \(F_S\) and \(F_{X_1}\) are continuous and strictly increasing, this amounts to \(F_S(2x)\le F_{X_1}(x)\) for \(x\ge0\), that is, \((1+2x)^{\alpha+1}(1+x)^{-\alpha}\le1+2(\alpha+1)x\). By Bernoulli's inequality, \(\bigl(1+\tfrac{x}{1+x}\bigr)^{\alpha}\le1+\tfrac{\alpha x}{1+x}\), so it suffices that \((1+2x)\bigl(1+(1+\alpha)x\bigr)\le\bigl(1+2(\alpha+1)x\bigr)(1+x)\), which reduces to \(\alpha x\ge0\). At \(\alpha=1\), one has \(F_S=F_{X_1}F_{X_2}\) identically, so the NSD frontier is attained, and \(\VaR_p[S]=(p+\sqrt p)/(1-p)\ge2p/(1-p)\).
\item[(2)] \emph{NSD without SD.} Let \(\XX\) be mutually exclusive \citep{Dhaene1999}, with \(\P(\XX=(2^k,0))=\P(\XX=(0,2^k))=2^{-(k+1)}\) for \(k\ge1\). Then \(F_{X_i}=\tfrac12\) on \([0,2)\) and \(F_{X_i}=1-2^{-(k+1)}\) on \([2^k,2^{k+1})\), while \(F_S=0\) on \([0,2)\) and \(F_S=1-2^{-k}\) on \([2^k,2^{k+1})\). Hence \(F_S(t)<\bigl(1-2^{-(k+1)}\bigr)^2=F_{X_1}(t)F_{X_2}(t)\) on \([2^k,2^{k+1})\), and \(\XX\) is NSD, as it must be, being counter-monotonic. With \(h=2^{-(k+2)}\),
\[
\phi_i(2^{k+1})-\phi_i(2^k)=2^k\log\frac{(1-h)^2}{1-2h}>0,
\]
so \(\Phi(2^k,2^k)<\Phi(2^{k+1},2^{k+1})\) and SD fails; the margins are moreover discontinuous. Nevertheless \(\VaR_p[S]=2>0=\VaR_p[X_1]+\VaR_p[X_2]\) for \(p\le\tfrac12\), and \(\VaR_p[S]=2^k=2\VaR_p[X_i]\) for \(1-2^{-(k-1)}<p\le1-2^{-k}\) and \(k\ge2\); thus \(\XX\) is VaR super-additive. Discrete margins are therefore compatible with global super-additivity.
\end{itemize}
\end{example}

Theorem~\ref{thm:Varsuperadditivity} is thus a sufficient criterion obtained by factorizing the exact condition of Lemma~\ref{lm:aggregationcriterion}, and by Proposition~\ref{prop:frechet-boundary} it is also necessary on the unit Fr\'echet boundary. The last subsection asks how much heaviness a super-additive vector must carry, with or without this factorization.

\subsection{How heavy must the margins be?}
\label{subsec:heaviness}
Remark~\ref{rm:super-additivityImamura} excludes genuine super-additivity when every component is integrable, and Proposition~\ref{prop:phi_mean_infinite} shows that each margin meeting the separable condition has a tail at least as heavy as a Pareto tail of index one. We now determine how much heaviness global super-additivity forces by itself, independently of any factorization, and which part of it simplex dominance enforces. The key is an identity that locates exactly where heaviness is needed.
\begin{lemma}[Alignment identity]
\label{lm:alignment}
Let $\XX$ satisfy Assumption~\ref{ass:zero-endpoint}. For every $q\in(0,1)$,
\begin{equation}
\label{eq:alignment}
\int_0^q\Bigl(\VaR_p[S]-\sum_{i=1}^n\VaR_p[X_i]\Bigr)\d p
=\inf_{A:\ \P(A)=q}\ \sum_{i=1}^n\mathcal E_i(A),
\end{equation}
where the infimum runs over events \(A\in\mathcal F\) and
\[
\mathcal E_i(A):=\E\bigl[X_i\mathbbm{1}_A\bigr]-\int_0^q\VaR_p[X_i]\,\d p\ \ge0 .
\]
\end{lemma}
\begin{proof}
Let $Z\ge0$ and $\P(A)=q$. Since $\P(Z>z,A)\ge\bigl(q-F_Z(z)\bigr)^+$, and $\VaR_p[Z]>z$ holds exactly when $p>F_Z(z)$,
\[
\E[Z\mathbbm{1}_A]=\int_0^\infty\P(Z>z,A)\,\d z\ \ge\ \int_0^\infty\bigl(q-F_Z(z)\bigr)^+\d z=\int_0^q\VaR_p[Z]\,\d p .
\]
Equality holds for $A=\{U_Z\le q\}$, where $U_Z$ is uniform on $(0,1)$ with $Z=\VaR_{U_Z}[Z]$ almost surely \citep{Ruschendorf2013}. Applying the inequality to each $X_i$ gives $\mathcal E_i(A)\ge0$. Applying it with equality to $Z=S$ and writing $\E[S\mathbbm{1}_A]=\sum_i\E[X_i\mathbbm{1}_A]$ gives \eqref{eq:alignment}.
\end{proof}

\begin{theorem}
\label{thm:intrinsic-heavy}
Let $\XX$ satisfy Assumption~\ref{ass:zero-endpoint} and be VaR super-additive, and set $L_i:=\limsup_{p\uparrow1}(1-p)\VaR_p[X_i]\in[0,\infty]$. If either
\begin{itemize}
\item[(i)] at most one component has infinite mean, or
\item[(ii)] there is an index $k$ with $L_k<\infty$ and $L_i=0$ for all $i\neq k$,
\end{itemize}
then $\XX$ is VaR additive and co-monotonic. Consequently, a vector that is VaR super-additive but not VaR additive has at least two components with infinite mean, and either at least two components with $L_i>0$ or at least one component with $L_i=\infty$.
\end{theorem}
\begin{proof}
Let $\Delta(p)=\VaR_p[S]-\sum_i\VaR_p[X_i]\ge0$ be the aggregation gap. The map $q\mapsto\int_0^q\Delta(p)\,\d p$ is non-decreasing and non-negative, so it suffices to show that it tends to $0$ as $q\uparrow1$. Then $\Delta=0$ almost everywhere, hence $\Delta\equiv0$ by left-continuity (Lemma~\ref{lm:VaRtoolkit}(i)), and co-monotonicity follows from Theorem~\ref{thm:VaRsubadditive}. Let $U_i$ be uniform with $X_i=\VaR_{U_i}[X_i]$ almost surely.

(i) Relabel so that $X_2,\dots,X_n$ are integrable, and take $A=\{U_1\le q\}$ in \eqref{eq:alignment}. Then $\mathcal E_1(A)=0$, while for $i\ge2$,
\[
\mathcal E_i(A)\ \le\ \E[X_i]-\int_0^q\VaR_p[X_i]\,\d p=\int_q^1\VaR_p[X_i]\,\d p\ \longrightarrow\ 0 .
\]

(ii) Relabel so that $k=1$. Fix $\eta,\varepsilon\in(0,1)$ and set $B=\{U_1>1-(1-\eta)\varepsilon\}\cup\bigcup_{i\ge2}\{U_i>1-\eta\varepsilon/(n-1)\}$. Then $\P(B)\le\varepsilon$, so there is $A\subseteq B^c$ with $\P(A)=1-\varepsilon$. As $X_1\ge0$ and $A\subseteq\{U_1\le1-(1-\eta)\varepsilon\}$,
\[
\mathcal E_1(A)\le\int_{1-\varepsilon}^{1-(1-\eta)\varepsilon}\VaR_p[X_1]\,\d p\le\frac{\eta}{1-\eta}\,(1-\eta)\varepsilon\,\VaR_{1-(1-\eta)\varepsilon}[X_1],
\]
and similarly, for $i\ge2$, $\mathcal E_i(A)\le\frac{n-1}{\eta}\,\frac{\eta\varepsilon}{n-1}\,\VaR_{1-\eta\varepsilon/(n-1)}[X_i]$. Letting $\varepsilon\downarrow0$ in \eqref{eq:alignment} gives $\limsup_{q\uparrow1}\int_0^q\Delta(p)\,\d p\le\frac{\eta}{1-\eta}L_1$, and letting $\eta\downarrow0$ completes the proof. The final assertion is the contrapositive of (i) and (ii).
\end{proof}

For every non-negative $X_i$ one has $L_i=\limsup_{t\to\infty}t\,\Fbar_{X_i}(t)$: if $p>F_{X_i}(s)$ then $\VaR_p[X_i]\ge s$, and if $x=\VaR_p[X_i]$ then $1-p\le\P(X_i\ge x)$. Theorem~\ref{thm:intrinsic-heavy} therefore requires either two tails with $L_i>0$ or one tail with $L_i=\infty$. Simplex dominance always produces the first configuration. With $\psi_i=-\phi_i$ as in Proposition~\ref{prop:phi_mean_infinite}, \(\Phi\) is SD if and only if $\sum_i\psi_i(x_i)\le\sum_i\psi_i(t)$ whenever $\sum_ix_i=t$.

\begin{proposition}
\label{prop:SD_heavy}
Let $\XX$ satisfy Assumption~\ref{ass:zero-endpoint}. If $\Phi$ in \eqref{eq:Phiexpression} is SD, then
\begin{equation}
\label{eq:SD_heavy}
\limsup_{t\to\infty}\sum_{i\neq k}\psi_i(t)\;\ge\;\sup_{s>0}\psi_j(s)\;>\;0,
\qquad\text{for all } j\neq k .
\end{equation}
Consequently:
\begin{itemize}
\item[(i)] at least two components have $L_i>0$, and in particular infinite mean;
\item[(ii)] if $n=2$, both components satisfy $L_j=\sup_{s>0}\psi_j(s)$.
\end{itemize}
\end{proposition}
\begin{proof}
Since $X_j$ is non-degenerate with lower endpoint $0$, there is $x_0>0$ with $0<F_{X_j}(x_0)<1$, so $\sup_s\psi_j(s)\ge\psi_j(x_0)>0$. Suppose, for a contradiction, that $\limsup_{t}\sum_{i\ne k}\psi_i(t)<\psi_j(x_0)-2\eta$ for some $\eta>0$. Then there is $T\ge x_0$ with $\sum_{i\ne k}\psi_i(t)\le\psi_j(x_0)-\eta$ for all $t\ge T$. Evaluating \eqref{eq:SD} at $x_j=x_0$, $x_k=t-x_0$ and $x_i=0$ otherwise gives
\[
\psi_j(x_0)+\psi_k(t-x_0)\;\le\;\psi_k(t)+\sum_{i\neq k}\psi_i(t)\;\le\;\psi_k(t)+\psi_j(x_0)-\eta,
\qquad t\ge T,
\]
that is, $\psi_k(t)\ge\psi_k(t-x_0)+\eta$. Iterating from $T$ yields $\psi_k(T+mx_0)\ge m\eta$ for every $m\in\NN$. Hence $\Lambda_k(T+mx_0)=\psi_k(T+mx_0)/(T+mx_0)$ stays bounded away from $0$ as $m\to\infty$, which contradicts $\Lambda_k(t)\to0$. Since $x_0$ and $\eta$ were arbitrary, \eqref{eq:SD_heavy} follows.

For (i), let $\mathcal I=\{i:\limsup_t\psi_i(t)>0\}$. If $|\mathcal I|\le1$, choose $k\in\mathcal I$ (any $k$ if $\mathcal I=\emptyset$) and $j\neq k$. Then $\sum_{i\ne k}\psi_i(t)\to0$, contradicting \eqref{eq:SD_heavy}; hence $|\mathcal I|\ge2$. Since $\Lambda_i(t)\sim\Fbar_{X_i}(t)$ as $t\to\infty$, we have $\limsup_t\psi_i(t)=\limsup_t t\,\Fbar_{X_i}(t)=L_i$; and $L_i>0$ forces $\E[X_i]=\infty$, because integrability would give $t\,\Fbar_{X_i}(t)\le\E\bigl[X_i\mathbbm{1}_{\{X_i>t\}}\bigr]\to0$. For (ii), take $n=2$ in \eqref{eq:SD_heavy}: then $\limsup_t\psi_j(t)\ge\sup_s\psi_j(s)$, and the reverse inequality is trivial.
\end{proof}

\begin{remark}
\label{rm:SD_heavy}
Proposition~\ref{prop:SD_heavy} is the non-separable counterpart of Proposition~\ref{prop:phi_mean_infinite}. Heaviness in the sense \(L_i>0\) is strictly stronger than an infinite mean -- the tail \(\Fbar(t)\sim1/(t\log t)\) is non-integrable, yet has \(L_i=0\) -- but simplex dominance imposes it on two components only; for \(n\ge3\) it leaves room for light components.
\end{remark}

\begin{example}
\label{example:mixed_means}
Let $X_1,X_2,X_3$ be independent, with $X_1$ uniform on $(0,\omega)$ and $F_{X_2}(x)=F_{X_3}(x)=\exp(-x^{-1/2})$. Then $\psi_2(x)=\psi_3(x)=\sqrt x$, while $\psi_1(x)=x\log(\omega/x)$ for $0<x<\omega$ and $\psi_1(x)=0$ for $x\ge\omega$. Writing $x=\omega v$,
\[
\sup_{x>0}\frac{\psi_1(x)}{\sqrt x}=\sqrt\omega\,\max_{v\in(0,1]}\sqrt v\,\log(1/v)=\frac{2\sqrt\omega}{e}.
\]
Hence, if $\omega\le\omega_0:=\bigl((2-\sqrt2)e/2\bigr)^2\approx0.634$, then $\psi_1(x_1)\le(2-\sqrt2)\sqrt{x_1}\le(2-\sqrt2)\sqrt t$. Combined with $\sqrt{x_2}+\sqrt{x_3}\le\sqrt{2(x_2+x_3)}\le\sqrt{2t}$, this gives
\[
\sum_{i=1}^3\psi_i(x_i)\;\le\;2\sqrt t\;\le\;\sum_{i=1}^3\psi_i(t),
\]
so $\Phi$ is SD. The vector is independent, hence NSD, and its margins are continuous. By Theorem~\ref{thm:Varsuperadditivity}, $\XX$ is VaR super-additive, although $X_1$ is bounded. Numerically, simplex dominance persists up to $\omega^\star\approx5.00$ and fails beyond, with the binding constraint at $t\approx\omega$ rather than at infinity. Condition \eqref{eq:SD_heavy} is therefore necessary but not sufficient: light components must be small relative to the heavy ones at every scale.
\end{example}

The second configuration, a single tail with \(L_i=\infty\), is out of reach of Theorem~\ref{thm:Varsuperadditivity}, but it is reached by factorizing through the lower Fr\'echet bound instead of the independence copula.
\begin{proposition}
\label{prop:Wcriterion}
Let $n=2$, let $\XX$ satisfy Assumption~\ref{ass:zero-endpoint}, and suppose $F_{\XX}(t,t)=\max\{F_{X_1}(t)+F_{X_2}(t)-1,0\}$ for all $t\ge0$; this holds in particular when $\XX$ is counter-monotonic. If
\begin{equation}
\label{eq:WQS}
\Fbar_{X_1}\bigl(t(p)\bigr)+\Fbar_{X_2}\bigl(t(p)\bigr)\ \ge\ 1-p,\qquad \forall p\in(0,1),
\end{equation}
where $t(p)=\VaR_p[X_1]+\VaR_p[X_2]$, then $\XX$ is VaR super-additive.
\end{proposition}
\begin{proof}
As in Proposition~\ref{prop:dependenceNSD}, $F_S(t)\le F_{\XX}(t,t)=\max\{F_{X_1}(t)+F_{X_2}(t)-1,0\}$. At $t=t(p)$, condition \eqref{eq:WQS} makes the right-hand side at most $p$, and Lemma~\ref{lm:aggregationcriterion}(ii) applies.
\end{proof}
\noindent Since $u+v-1\le uv$ on $[0,1]^2$, condition \eqref{eq:WQS} is implied by \eqref{eq:QS} for $n=2$ and is strictly weaker.
\begin{example}
\label{example:configII}
Let $\alpha\in(0,\tfrac16]$, and let $(X_1,X_2)$ be counter-monotonic with
\[
\Fbar_{X_1}(x)=(1+x)^{-\alpha},
\qquad
\Fbar_{X_2}(z)=\Bigl[(1+z)\bigl(1+\log(1+z)\bigr)^{1/2}\Bigr]^{-1},
\qquad x,z\ge0 .
\]
The substitution $u=\log(1+z)$ gives $\E[X_2]=\int_0^\infty(1+u)^{-1/2}\,\d u=\infty$, and clearly $\E[X_1]=\infty$. Moreover $L_1=\infty$, whereas $z\,\Fbar_{X_2}(z)\to0$, so $L_2=0$. By Proposition~\ref{prop:SD_heavy}(ii), $\Phi$ is not SD, and Theorem~\ref{thm:Varsuperadditivity} does not apply, although $\XX$ is NSD.

We verify \eqref{eq:WQS}. Fix $p$, put $\varepsilon=1-p$, $y=\VaR_p[X_2]$ and $\ell=\log(1+y)$. Then $(1+y)(1+\ell)^{1/2}=1/\varepsilon$, so $\log(1/\varepsilon)\le\tfrac32\ell$ and $\varepsilon y\le(1+\ell)^{-1/2}\le1$. Put $z=\varepsilon^{1/\alpha}y\le1$. Since $1+t(p)=\varepsilon^{-1/\alpha}(1+z)$, convexity of $u\mapsto(1+u)^{-\alpha}$ gives
\[
\Fbar_{X_1}\bigl(t(p)\bigr)=\varepsilon(1+z)^{-\alpha}\ \ge\ \varepsilon-\alpha\,\varepsilon^{1/\alpha}(1+\ell)^{-1/2}.
\]
Further, $z\le1$ gives $\log(1+t(p))\le\alpha^{-1}\log(1/\varepsilon)+\log2$, whence
\[
\Fbar_{X_2}\bigl(t(p)\bigr)\ \ge\ \frac{\varepsilon^{1/\alpha}}{2\bigl(1+\log2+\alpha^{-1}\log(1/\varepsilon)\bigr)^{1/2}} .
\]
Consequently $\Fbar_{X_1}(t(p))+\Fbar_{X_2}(t(p))>\varepsilon$ as soon as $1+\ell>4\alpha^2(1+\log2)+4\alpha\log(1/\varepsilon)$. This follows from $1+\ell>4\alpha^2(1+\log2)+6\alpha\ell$, which holds because $6\alpha\le1$ and $4\alpha^2(1+\log2)<1$. By Proposition~\ref{prop:Wcriterion}, $\XX$ is VaR super-additive. Being counter-monotonic with non-degenerate margins, it is not co-monotonic, hence not VaR additive by Theorem~\ref{thm:VaRsubadditive}, and by Corollary~\ref{coro:genericsuperadditivity} it is strictly super-additive on a set of levels of positive measure. In high-precision arithmetic, \eqref{eq:WQS} continues to hold for all $\alpha\le0.8$ and first fails near $\alpha\approx0.85$; the relative slack decays like $(1-p)^{1/\alpha-1}$, so double-precision evaluations are unreliable at high levels.
\end{example}

\begin{remark}
\label{rm:intrinsic-vs-SD}
Theorem~\ref{thm:intrinsic-heavy}(i) strengthens Remark~\ref{rm:super-additivityImamura} from all components integrable to all but one. NSD imposes no heaviness at all, since independent vectors are NSD whatever their margins; within Theorem~\ref{thm:Varsuperadditivity} the heaviness requirement is therefore carried entirely by SD, which forces two components with \(L_i>0\), and Example~\ref{example:configII} shows that this is an artefact of factorizing through the independence copula. Every configuration permitted by Theorem~\ref{thm:intrinsic-heavy} occurs: two tails with \(L_i>0\) in the opening example of this section, exactly two infinite-mean components alongside a bounded one in Example~\ref{example:mixed_means}, and a single tail with \(L_i=\infty\) in Example~\ref{example:configII}. The theorem is therefore sharp.
\end{remark}

%%%%%%%%%%%%%%%%%%%%%%%%%%%%%%%%%%%%%%%%%%%%%%%%%%%%%%%%%%%%%%%%%%%%%%%%%%%%%%%%%%%%%%%%%%%%%%%%%%%%%
\section{Risks with general endpoints}
\label{sec:Generalizations}
In practice supports are often shifted or capped: under a franchise deductible \citep{Klugman2019} a loss lives on \([a_i,\infty)\), whereas capped recoveries and positions with bounded upside live on \((-\infty,b_i]\). A single reflection principle transfers all previous results to these settings. We consider random variables with finite lower endpoints
\[
a_i = \sup\{x \in \mathbb{R} : F_{X_i^{a_i}}(x) \le 0\} > -\infty,
\qquad \forall i \in \{1,\dots,n\},
\]
and, symmetrically, with arbitrary finite upper endpoints
\[
b_i = \inf\{x \in \mathbb{R} : F_{X_i^{b_i}}(x) \ge 1\} < \infty, 
\qquad \forall i \in \{1,\dots,n\}.
\]
We denote the corresponding shifted and reflected random vectors and their aggregate sums by
\[
\XX^{\aa}=(X_1^{a_1},\dots,X_n^{a_n}),\qquad \XX^{\bb}=(X_1^{b_1},\dots,X_n^{b_n}),\qquad \aa=(a_1,\dots,a_n),\ \bb=(b_1,\dots,b_n)\in\RR^n,
\]
\[
S^{\aa}=\sum_{i=1}^n X_i^{a_i},
\qquad 
S^{\bb}=\sum_{i=1}^n X_i^{b_i},
\qquad
a_+=\sum_{i=1}^n a_i,
\qquad
b_+=\sum_{i=1}^n b_i,
\]
where $\XX^{\mathbf{0}}=\XX$ recovers the case of risks supported on $[0,\infty)$. Throughout the section the associated \emph{reduced vector} is
\begin{equation}
\label{eq:reduction}
\XX=(X_1,\dots,X_n),
\qquad
X_i:=X_i^{a_i}-a_i
\quad\text{respectively}\quad
X_i:=b_i-X_i^{b_i},
\end{equation}
and each choice makes \(\XX\) satisfy Assumption~\ref{ass:zero-endpoint}; every hypothesis below is the exact translate of one on \(\XX\).

It is convenient to measure aggregation behavior by a single function. For a random vector $\ZZ=(Z_1,\dots,Z_n)$ with aggregate $S_{\ZZ}$, define the \emph{aggregation gap}
\begin{equation}
\label{eq:gap}
\Delta_{\ZZ}(p):=\VaR_p[S_{\ZZ}]-\sum_{i=1}^n\VaR_p[Z_i],
\qquad
\Delta_{\ZZ}^{+}(p):=\VaRplus_p[S_{\ZZ}]-\sum_{i=1}^n\VaRplus_p[Z_i],
\qquad p\in(0,1),
\end{equation}
so that $\ZZ$ is VaR super-additive when $\Delta_{\ZZ}\ge0$ throughout $(0,1)$, VaR sub-additive when $\Delta_{\ZZ}\le0$, and VaR additive when $\Delta_{\ZZ}\equiv0$.

\subsection{The reflection principle}
\label{subsec:reflection}
All results of this section follow from the next proposition.
\begin{proposition}[Translation and reflection principle]
\label{prop:reflection}
Let $\XX$ be the reduced vector of \eqref{eq:reduction}. Then, for every $p\in(0,1)$,
\begin{equation}
\label{eq:gap-shift}
\Delta_{\XX^{\aa}}(p)=\Delta_{\XX}(p),
\end{equation}
\begin{equation}
\label{eq:gap-reflect}
\Delta_{\XX^{\bb}}(p)=-\,\Delta_{\XX}^{+}(1-p).
\end{equation}
Consequently:
\begin{itemize}
\item[(i)] $\XX^{\aa}$ is VaR sub-additive (respectively super-additive, additive) if and only if $\XX$ is VaR sub-additive (respectively super-additive, additive);
\item[(ii)] $\XX^{\bb}$ is VaR sub-additive (respectively super-additive, additive) if and only if $\XX$ is VaR super-additive (respectively sub-additive, additive);
\item[(iii)] $\XX^{\aa}$ is co-monotonic if and only if $\XX$ is, and likewise for $\XX^{\bb}$.
\end{itemize}
\end{proposition}
\begin{proof}
Since $X_i^{a_i}=a_i+X_i$ and $S^{\aa}=S+a_+$, translation equivariance (Lemma~\ref{lm:VaRtoolkit}(ii)) gives $\VaR_p[S^{\aa}]=a_++\VaR_p[S]$ and $\sum_{i}\VaR_p[X_i^{a_i}]=a_++\sum_i\VaR_p[X_i]$; subtracting yields \eqref{eq:gap-shift}, and (i) follows at once.

Since $X_i^{b_i}=b_i-X_i$ and $S^{\bb}=b_+-S$, the reflection identity of Lemma~\ref{lm:VaRtoolkit}(iv) gives $\VaR_p[S^{\bb}]=b_+-\VaRplus_{1-p}[S]$ and $\sum_i\VaR_p[X_i^{b_i}]=b_+-\sum_i\VaRplus_{1-p}[X_i]$; subtracting yields \eqref{eq:gap-reflect}. Hence $\Delta_{\XX^{\bb}}\ge0$ on $(0,1)$ if and only if $\Delta^{+}_{\XX}\le0$ on $(0,1)$, which by Lemma~\ref{lm:VaRversion} is equivalent to $\Delta_{\XX}\le0$ on $(0,1)$; the other two cases are identical, and (ii) follows. Statement (iii) holds because co-monotonicity means that the support is a chain in $\RR^n$ (Remark~\ref{rm:comonotonicity-facts}(b)), a property invariant under translation and under simultaneous reflection of all coordinates.
\end{proof}

Thus shifting leaves the aggregation profile unchanged, while reflection reverses it and exchanges \(p\) with \(1-p\). The impossibility theorem transfers verbatim.
\begin{proposition}
\label{prop:XaXbImpossibility}
Let $a_i$ be the finite lower endpoint of $X_i^{a_i}$ and $b_i$ the finite upper endpoint of $X_i^{b_i}$, for all $i\in\{1,\dots,n\}$. Then:
\begin{itemize}
\item[(i)] $\XX^{\aa}$ is VaR sub-additive $\iff$ $\XX^{\aa}$ is VaR additive $\iff$ $\XX^{\aa}$ is co-monotonic;
\item[(ii)] $\XX^{\bb}$ is VaR super-additive $\iff$ $\XX^{\bb}$ is VaR additive $\iff$ $\XX^{\bb}$ is co-monotonic.
\end{itemize}
In particular, strict VaR sub-additivity is unattainable for risks that are bounded below, and strict VaR super-additivity is unattainable for risks that are bounded above.
\end{proposition}
\noindent The proof is given in Appendix~\ref{app:XaXbImpossibility}.

A floor rules out diversification benefits at all levels, and a ceiling rules out diversification penalties; neither statement requires moments. When both endpoints are finite, the two collide.
\begin{corollary}
\label{coro:compact}
Let $\XX^{\aa,\bb}=(X_1^{a_1,b_1},\dots,X_n^{a_n,b_n})$ be a random vector whose $i$-th component has finite lower endpoint $a_i$ and finite upper endpoint $b_i$. Then VaR sub-additivity, VaR super-additivity, VaR additivity and co-monotonicity of $\XX^{\aa,\bb}$ are all equivalent. Consequently, a compactly supported vector that is not co-monotonic is neither VaR sub-additive nor VaR super-additive: its aggregation gap $\Delta_{\XX^{\aa,\bb}}$ changes sign on $(0,1)$.
\end{corollary}
\begin{proof}
Apply Proposition~\ref{prop:XaXbImpossibility}(i) and (ii) to the same vector; each yields equivalence with co-monotonicity. If the vector is not co-monotonic, both global inequalities fail, so $\Delta_{\XX^{\aa,\bb}}$ is somewhere strictly positive and somewhere strictly negative.
\end{proof}

\subsection{Super-additivity above a floor and sub-additivity below a ceiling}
\label{subsec:endpoint-conditions}
For compact supports, Corollary~\ref{coro:compact} also follows from \citet{Imamura2025}, since such risks are integrable. The two remaining cells of Table~\ref{tab:roadmap} are opened by transporting Section~\ref{sec:VaR_super-additivity}. In the reflected case closed events are needed, since \(\{S^{\bb}\ge b_+-t\}=\{S\le t\}\) whereas \(\{S^{\bb}>b_+-t\}=\{S<t\}\).
\begin{proposition}
\label{prop:VaR_super-extension}
~
\begin{itemize}
\item[(i)]
Suppose $a_i$ is the finite lower endpoint of $X_i^{a_i}$, the marginal CDFs $F_{X_i^{a_i}}$ are continuous, and
\begin{equation}
\label{eq:XaNSD}
F_{S^{\aa}}(t+a_+) \le \prod_{i=1}^n F_{X_i^{a_i}}(t+a_i),
\qquad \forall t \in [0,\infty),
\end{equation}
and suppose further that the function
\begin{equation}
\label{eq:XaSD}
\Phi^{\aa}(x_1,\dots,x_n)
    = \sum_{i=1}^n x_i \log F_{X_i^{a_i}}(x_i+a_i),\qquad x_i\in[0,\infty),
\end{equation}
is SD. Then $\XX^{\aa}$ is VaR super-additive.

\item[(ii)]
Suppose $b_i$ is the finite upper endpoint of $X_i^{b_i}$, the marginal DDFs $\Fbar_{X_i^{b_i}}$ are continuous, and
\begin{equation}
\label{eq:XbNSD}
\P\bigl(S^{\bb} \ge b_+ - t\bigr)
    \le \prod_{i=1}^n \P\bigl(X_i^{b_i} \ge b_i - t\bigr),
\qquad \forall t \in [0,\infty),
\end{equation}
and suppose further that the function
\begin{equation}
\label{eq:XbSD}
\Phi^{\bb}(x_1,\dots,x_n)
    = \sum_{i=1}^n x_i \log \Fbar_{X_i^{b_i}}(b_i-x_i),\qquad x_i\in[0,\infty),
\end{equation}
is SD. Then $\XX^{\bb}$ is VaR sub-additive.
\end{itemize}
\end{proposition}
\noindent The complete proof appears in Appendix~\ref{app:VaR_super-extension}.

\begin{remark}
\label{rm:closed-events}
Under continuous margins, the right-hand side of \eqref{eq:XbNSD} equals \(\prod_i\Fbar_{X_i^{b_i}}(b_i-t)\); the left-hand side, however, is not \(\Fbar_{S^{\bb}}(b_+-t)=F_S(t^-)\) unless \(F_S\) is continuous. The weaker left-limit condition in fact suffices, by Lemma~\ref{lm:aggregationcriterion}(i).
\end{remark}

The sufficient conditions carry over, with upper orthants in the reflected case.
\begin{definition}
\label{def:NUOD}
A random vector $\ZZ=(Z_1,\dots,Z_n)$ is \emph{negative upper orthant dependent} (NUOD) \citep{Block1982a,Block1985,Joe1997} if
\[
\P\bigl(Z_1\ge z_1,\dots,Z_n\ge z_n\bigr)\le \prod_{i=1}^n\P\bigl(Z_i\ge z_i\bigr),
\qquad \forall (z_1,\dots,z_n)\in\RR^n .
\]
\end{definition}
\begin{proposition}
\label{prop:VaR_super-extension-sufficient}
~
\begin{itemize}
\item[(i)] If $\XX^{\aa}$ is NLOD with continuous $F_{X_i^{a_i}}$, and if each function appearing in 
Equation~\eqref{eq:XaSD},
\[
\phi_i^{a_i}(x_i)=x_i\log F_{X_i^{a_i}}(x_i+a_i),\qquad x_i\in[0,\infty),
\]
is non-increasing, then $\XX^{\aa}$ is VaR super-additive.

\item[(ii)] If $\XX^{\bb}$ is NUOD with continuous $\Fbar_{X_i^{b_i}}$, and if each function appearing in Equation~\eqref{eq:XbSD},
\[
\phi_i^{b_i}(x_i)=x_i\log \Fbar_{X_i^{b_i}}(b_i-x_i),\qquad x_i\in[0,\infty),
\]
is non-increasing, then $\XX^{\bb}$ is VaR sub-additive.
\end{itemize}
In either case it suffices that the orthant condition hold along the shifted diagonal $(t+a_1,\dots,t+a_n)$, respectively $(b_1-t,\dots,b_n-t)$, for $t\in[0,\infty)$.
\end{proposition}
\noindent The result is proved in Appendix~\ref{app:VaR_super-extension-sufficient}.

\begin{remark}
\label{rm:translation-invariance}
Since \(F_{X_i^{a_i}}(x+a_i)=F_{X_i}(x)\) and, for continuous margins, \(\Fbar_{X_i^{b_i}}(b_i-x)=F_{X_i}(x)\), both \(\Phi^{\aa}\) and \(\Phi^{\bb}\) coincide with the aggregator \eqref{eq:Phiexpression} of the reduced vector. The margin condition is therefore endpoint-free: every law in Table~\ref{tab:margins} remains admissible under a franchise deductible of any size, and only the dependence condition changes, from lower to upper orthants.
\end{remark}

Reflecting the opening example of Section~\ref{sec:VaR_super-additivity} gives an explicit sub-additive vector.
\begin{example}
\label{example:mirrored}
Let \(\XX^{\bb}=(b_1-X,\,b_2-1/X)\), with \(X\) Pareto (II) of unit scale and shape \(\tfrac12\); this models, for instance, two positions with capped upside and heavy-tailed downside. The reduced vector \((X,1/X)\) is counter-monotonic, hence NLOD, so \(\XX^{\bb}\) is NUOD, and its reduced margins (Pareto (II) and Beta prime, both of shape \(\tfrac12\)) appear in Table~\ref{tab:margins}. By Proposition~\ref{prop:VaR_super-extension-sufficient}(ii), \(\XX^{\bb}\) is VaR sub-additive. All distribution functions involved are continuous and strictly increasing on their supports, so \(\Delta^+_{\XX}=\Delta_{\XX}\), and \eqref{eq:gap-reflect} together with \eqref{eq:opening-gap} gives the exact diversification benefit
\[
\Delta_{\XX^{\bb}}(p)=-\frac{2}{p\,(2-p)^2}<0,\qquad p\in(0,1),
\]
the mirror image of \eqref{eq:opening-gap}.
\end{example}

\subsection{Closure under transformations}
\label{subsec:endpoint-closure}
Proposition~\ref{prop:VaR_X_convex} transfers with star-shapedness measured from the endpoint; reflection turns convexity into concavity.
\begin{proposition}
\label{prop:XaXb_Xi_convex}
~
\begin{itemize}
\item[(i)] Suppose $\XX^{\aa}$ is NLOD with continuous margins $F_{X_i^{a_i}}$, and assume that each $\phi_i^{a_i}$ in Equation \eqref{eq:XaSD} is non-increasing. Let 
\[
\widetilde{\XX}^{\aa}=\left(\widetilde{X}_1^{a_1},\dots,\widetilde{X}_n^{a_n}\right),
\qquad \widetilde{X}_i^{a_i}=\xi_i\bigl(X_i^{a_i}\bigr),
\quad \xi_i:[a_i,\infty)\to[a_i,\infty).
\]
If each $\xi_i$ is strictly increasing, satisfies $\xi_i(a_i)=a_i$, and is \emph{star-shaped at} $a_i$, that is
\[
u\longmapsto \frac{\xi_i(u)-a_i}{u-a_i}
\quad\text{is non-decreasing on } (a_i,\infty),
\]
then $\widetilde{\XX}^{\aa}$ is VaR super-additive. In particular this holds when $\xi_i$ is strictly increasing, \emph{convex}, and satisfies $\xi_i(a_i)=a_i$.

\item[(ii)] Assume $\XX^{\bb}$ is NUOD with continuous margins $\Fbar_{X_i^{b_i}}$, and suppose that each $\phi_i^{b_i}$ in Equation \eqref{eq:XbSD} is non-increasing. Define
\[
\widetilde{\XX}^{\bb}=\left(\widetilde{X}_1^{b_1},\dots,\widetilde{X}_n^{b_n}\right),
\qquad \widetilde{X}_i^{b_i}=\xi_i\bigl(X_i^{b_i}\bigr),
\quad \xi_i:(-\infty,b_i]\to(-\infty,b_i].
\]
If each $\xi_i$ is strictly increasing, satisfies $\xi_i(b_i)=b_i$, and is \emph{anti-star-shaped at} $b_i$, that is
\[
u\longmapsto \frac{b_i-\xi_i(u)}{b_i-u}
\quad\text{is non-increasing on } (-\infty,b_i),
\]
then $\widetilde{\XX}^{\bb}$ is VaR sub-additive. In particular this holds when $\xi_i$ is strictly increasing, \emph{concave}, and satisfies $\xi_i(b_i)=b_i$.
\end{itemize}
\end{proposition}
\noindent The proof is provided in Appendix~\ref{app:XaXb_Xi_convex}.

\begin{remark}
\label{rm:concavity}
On the reduced vector, \(\xi_i\) induces \(\zeta_i(x)=\xi_i(a_i+x)-a_i\), respectively \(\zeta_i(x)=b_i-\xi_i(b_i-x)\), and Proposition~\ref{prop:VaR_X_convex} requires \(\zeta_i\) to be star-shaped. In the reflected case \(x\mapsto\xi_i(b_i-x)\) is concave exactly when \(\xi_i\) is, so it is concavity of \(\xi_i\) that makes \(\zeta_i\) convex. Convexity does not suffice: the vector \(\XX^{\bb}=(-X,-1/X)\) of Example~\ref{example:mirrored} with \(b_1=b_2=0\) is VaR sub-additive, but under the convex map \(\xi_i(u)=e^{u}-1\) the reduced components become \(1-e^{-X}\) and \(1-e^{-1/X}\), which are compactly supported and counter-monotonic. By Corollary~\ref{coro:compact} and Proposition~\ref{prop:reflection}(ii), the transformed vector is not VaR sub-additive.
\end{remark}

\section{Conclusions}
\label{sec:conclusions}
This paper studied the aggregation of VaR for one-sided risks when diversification is rewarded, or penalized, at every probability level simultaneously; Table~\ref{tab:roadmap} summarizes the answer.

For risks bounded below, sub-additivity at all levels is equivalent to additivity and to co-monotonicity, with no moment condition. The key is the domination \(0\le X_i\le S\), which substitutes one-sidedness for integrability. Consequently, every non-co-monotonic portfolio of such risks is strictly super-additive on a set of levels of positive measure, and sub-additivity of a portfolio forces additivity of all its sub-portfolios.

Global super-additivity is the stochastic dominance of the aggregate over its co-monotonic counterpart, equivalently the exact condition \(F_S\bigl((\sum_i\VaR_p[X_i])^-\bigr)\le p\). We factorized it through the maximum of independent copies of the margins into negative simplex dependence and simplex dominance; for continuous margins the latter holds when each reciprocal risk is DFRA, that is, when each strictly increasing margin is a star-shaped transformation of a unit Fr\'echet law. The criterion allows heterogeneous margins and dependence beyond NLOD, recovers the super-additivity results of \citet{Chen2026}, bounds the diversification penalty from below by a quantity that depends on the margins only, and is stable under star-shaped transformations. Neither condition can be dropped and neither is necessary, but for unit Fr\'echet margins the factorization is exact and super-additivity is equivalent to NSD. Independently of any factorization, genuine global super-additivity requires at least two infinite-mean components, and this requirement is sharp. A reflection principle extends all results to general endpoints: a floor rules out global sub-additivity, and a ceiling rules out global super-additivity.

The practical message is threefold. For lower-bounded, heavy-tailed lines there is no middle ground between genuine diversification and lockstep. The failure of diversification can be diagnosed from primitives -- an orthant check of the copula along the diagonal and an ageing check on each margin -- precisely where expected shortfall is infinite and uninformative. And positive dependence does not rescue diversification: a catastrophe shock common to all lines, which creates tail dependence, can erode the penalty but not reverse it.

Several questions remain open: a workable characterization of global super-additivity; statistical tests of NSD; consequences for risk sharing and capital allocation among super-additive lines; and the two-sided case, where our truncation argument is unavailable. There, the finite-mean hypothesis of \citet{Imamura2025} cannot simply be dropped: for independent standard Cauchy variables $X_1+X_2\eqd2X_1$, so VaR is additive at every level although the vector is not co-monotonic.

%%%%%%%%%%%%%%%%%%%%%%%%%%%%%%%%%%%%%%%%%%%%%%%%%%%%%%%%%%%%%%%%%%%%%%%%%%%%%%%%%%%%%%%%%%%%%%%%%%%%%

%%%%%%%%%%%%%%%%%%%%%%%%%%%%%%%%%%%%%%%%%%%%%%%%%%%%%%%%%%%%%%%%%%%%%%%%%%%%%%%%%%%%%%%%%%%%%%%%%%%%%
\newpage
\newpage
\begin{appendices}
\renewcommand{\thesection}{\Alph{section}}
\section{Proofs of Section \ref{sec:Generalizations}}
%\addcontentsline{toc}{section}{Proofs}
\renewcommand{\thesubsection}{\thesection.\roman{subsection}}
\subsection{Proof of Proposition \ref{prop:XaXbImpossibility}}
\label{app:XaXbImpossibility}
\begin{proof}
Let $\XX$ be the reduced vector of \eqref{eq:reduction}. Because $a_i$ is the lower endpoint of $X_i^{a_i}$ in part~(i), and $b_i$ the upper endpoint of $X_i^{b_i}$ in part~(ii), the vector $\XX$ is non-degenerate with all lower endpoints at the origin and therefore satisfies Assumption~\ref{ass:zero-endpoint}. Theorem~\ref{thm:VaRsubadditive} consequently applies to $\XX$, and Proposition~\ref{prop:reflection} transports each of its three equivalent statements.

\medskip
\noindent{(i)}  
By Proposition~\ref{prop:reflection}(i),
\[
\XX^{\aa}\text{ is VaR sub-additive} \iff \XX\text{ is VaR sub-additive},
\]
and by Theorem~\ref{thm:VaRsubadditive} the right-hand statement is equivalent to $\XX$ being VaR additive and to $\XX$ being co-monotonic. Applying Proposition~\ref{prop:reflection}(i) again to the additivity statement, and Proposition~\ref{prop:reflection}(iii) to the co-monotonicity statement, yields
\[
\XX^{\aa}\text{ is VaR sub-additive} 
\iff 
\XX^{\aa}\text{ is VaR additive}
\iff 
\XX^{\aa}\text{ is co-monotonic}.
\]

\medskip
\noindent{(ii)}  
By Proposition~\ref{prop:reflection}(ii), which is where the reflection identity $\VaR_p[b_i-X_i]=b_i-\VaRplus_{1-p}[X_i]$ of Lemma~\ref{lm:VaRtoolkit}(iv) and the version-invariance of Lemma~\ref{lm:VaRversion} are used,
\[
\XX^{\bb}\text{ is VaR super-additive} 
\iff 
\XX\text{ is VaR sub-additive}.
\]
Theorem~\ref{thm:VaRsubadditive} again identifies the right-hand statement with VaR additivity and with co-monotonicity of $\XX$, and Proposition~\ref{prop:reflection}(ii)--(iii) transports both back, giving
\[
\XX^{\bb}\text{ is VaR super-additive}
\iff 
\XX^{\bb}\text{ is VaR additive}
\iff
\XX^{\bb}\text{ is co-monotonic}.
\]

\medskip
The final assertions of the proposition follow because a vector that is VaR additive at every level has aggregation gap identically zero, so neither strict inequality in \eqref{eq:VaRsubsuper} can hold anywhere.
\end{proof}
\subsection{Proof of Proposition \ref{prop:VaR_super-extension}}
\label{app:VaR_super-extension}
\begin{proof}
Throughout, $\XX$ denotes the reduced vector of \eqref{eq:reduction}, which satisfies Assumption~\ref{ass:zero-endpoint}; continuity of each $F_{X_i}$ follows from that of $F_{X_i^{a_i}}$ in part~(i) and from that of $\Fbar_{X_i^{b_i}}$ in part~(ii), since translation and reflection preserve continuity.

\medskip
\noindent{(i)}  
Since $\XX^{\aa}=\aa+\XX$ and $S^{\aa}=S+a_+$, we have the exact identities
\[
F_{S^{\aa}}(t+a_+)=F_S(t),
\qquad 
F_{X_i^{a_i}}(x_i+a_i)=F_{X_i}(x_i),
\qquad \forall t,x_i\in\RR .
\]
Condition \eqref{eq:XaNSD} is therefore literally the statement
\[
F_{S}(t) \le \prod_{i=1}^n F_{X_i}(t),\qquad \forall t\in[0,\infty),
\]
that is, $\XX$ is NSD. Moreover the same identities give $\Phi^{\aa}=\Phi$ as functions on $[0,\infty)^n$, where $\Phi$ is the aggregator of \eqref{eq:Phiexpression}, so the hypothesis that $\Phi^{\aa}$ be SD is the hypothesis that $\Phi$ be SD. Theorem~\ref{thm:Varsuperadditivity} now applies and $\XX$ is VaR super-additive; Proposition~\ref{prop:reflection}(i) then gives that $\XX^{\aa}$ is VaR super-additive.

\medskip
\noindent{(ii)}  
Since $\XX^{\bb}=\bb-\XX$ and $S^{\bb}=b_+-S$, the events $\{S^{\bb}\ge b_+-t\}$ and $\{S\le t\}$ coincide, as do $\{X_i^{b_i}\ge b_i-x_i\}$ and $\{X_i\le x_i\}$. Hence
\[ 
\P\bigl(S^{\bb}\ge b_+-t\bigr)=F_S(t),
\qquad
\P\bigl(X_i^{b_i}\ge b_i-x_i\bigr)=F_{X_i}(x_i),
\qquad \forall t,x_i\in\RR,
\]
with no continuity required, and condition \eqref{eq:XbNSD} states exactly that $\XX$ is NSD. Continuity of $\Fbar_{X_i^{b_i}}$ further gives $\Fbar_{X_i^{b_i}}(b_i-x_i)=\P(X_i^{b_i}\ge b_i-x_i)=F_{X_i}(x_i)$, so $\Phi^{\bb}=\Phi$ and the hypothesis that $\Phi^{\bb}$ be SD is again the hypothesis that $\Phi$ be SD. By Theorem~\ref{thm:Varsuperadditivity} the reduced vector $\XX$ is VaR super-additive, and Proposition~\ref{prop:reflection}(ii) concludes that $\XX^{\bb}$ is VaR sub-additive.
\end{proof}
\subsection{Proof of Proposition \ref{prop:VaR_super-extension-sufficient}}
\label{app:VaR_super-extension-sufficient}
\begin{proof}
\begin{itemize}
\item[(i)]  
We verify the condition in Equation~\eqref{eq:XaNSD}.  
Since $\XX^{\aa}$ is NLOD,
\[
F_{\XX^{\aa}}(x_1,\dots,x_n)
    \le \prod_{i=1}^n F_{X_i^{a_i}}(x_i),
    \qquad \forall x_i\in[a_i,\infty).
\]
Because $X_i^{a_i}\ge a_i$ almost surely, the event $\{S^{\aa}\le t+a_+\}$ entails $\sum_{j=1}^n\bigl(X_j^{a_j}-a_j\bigr)\le t$ with all summands non-negative, hence $X_i^{a_i}\le t+a_i$ for every $i$; equivalently, the $n$-simplex with apex $(a_1,\dots,a_n)$ and vertices
\[
\{ (t+a_1,a_2,\dots,a_n), (a_1,t+a_2,\dots,a_n), \dots, (a_1,a_2,\dots,t+a_n) \}
\]
is contained in the $n$-box $[a_1,t+a_1]\times\dots\times[a_n,t+a_n]$. Therefore
\[
F_{S^{\aa}}(t+a_+) 
    \;\le\; F_{\XX^{\aa}}(t+a_1,\dots,t+a_n)
    \;\le\; \prod_{i=1}^n F_{X_i^{a_i}}(t+a_i),
    \qquad \forall t\in[0,\infty),
\]
which is Equation~\eqref{eq:XaNSD}. Only the values of $F_{\XX^{\aa}}$ along the shifted diagonal $(t+a_1,\dots,t+a_n)$ are used, so NLOD along that diagonal suffices.

If in addition each $\phi_i^{a_i}$ is non-increasing, then $\Phi^{\aa}=\sum_i\phi_i^{a_i}$ is SD by Proposition~\ref{prop:Phimargins}. Combined with the continuity of each $F_{X_i^{a_i}}$, Proposition~\ref{prop:VaR_super-extension}(i) yields that $\XX^{\aa}$ is VaR super-additive.

\item[(ii)]
The argument mirrors part~(i).  
By Definition~\ref{def:NUOD},
\[
\P\bigl(X_1^{b_1}\ge x_1,\dots,X_n^{b_n}\ge x_n\bigr)
    \le \prod_{i=1}^n \P\bigl(X_i^{b_i}\ge x_i\bigr),
    \qquad \forall x_i\in(-\infty,b_i].
\]
Because $X_i^{b_i}\le b_i$ almost surely, the event $\{S^{\bb}\ge b_+-t\}$ entails $\sum_{j=1}^n\bigl(b_j-X_j^{b_j}\bigr)\le t$ with all summands non-negative, hence $X_i^{b_i}\ge b_i-t$ for every $i$; equivalently the reversed $n$-simplex with apex $(b_1,\dots,b_n)$ and vertices
\[
\{(b_1-t,b_2,\dots,b_n), (b_1,b_2-t,\dots,b_n), \dots, (b_1,b_2,\dots,b_n-t)\}
\]
is contained in the $n$-box $[b_1-t,b_1]\times\dots\times[b_n-t,b_n]$. Setting $x_i=b_i-t$ with $t\in[0,\infty)$ therefore gives
\[
\P\bigl(S^{\bb}\ge b_+-t\bigr)
    \;\le\; \P\bigl(X_1^{b_1}\ge b_1-t,\dots,X_n^{b_n}\ge b_n-t\bigr)
    \;\le\; \prod_{i=1}^n \P\bigl(X_i^{b_i}\ge b_i-t\bigr),
\]
for every $t\in[0,\infty)$, which is Equation~\eqref{eq:XbNSD}; again only the shifted diagonal $(b_1-t,\dots,b_n-t)$ is used.

Finally, if each $\phi_i^{b_i}$ is non-increasing, then $\Phi^{\bb}$ is SD by Proposition~\ref{prop:Phimargins}, and together with the continuity of each $\Fbar_{X_i^{b_i}}$, Proposition~\ref{prop:VaR_super-extension}(ii) implies that $\XX^{\bb}$ is VaR sub-additive.
\end{itemize}
\end{proof}
\subsection{Proof of Proposition \ref{prop:XaXb_Xi_convex}}
\label{app:XaXb_Xi_convex}
\begin{proof}
Both parts reduce to Proposition~\ref{prop:VaR_X_convex} applied to the reduced vector $\XX$ of \eqref{eq:reduction}, which satisfies Assumption~\ref{ass:zero-endpoint} and has continuous margins.

(i) Here $\XX=\XX^{\aa}-\aa$ is NLOD, because orthant inequalities are invariant under translation, and $\phi_i=\phi_i^{a_i}$ by Remark~\ref{rm:translation-invariance}, so each $\phi_i$ is non-increasing. Put $\zeta_i(x)=\xi_i(a_i+x)-a_i$ for $x\ge0$. Then $\zeta_i$ is strictly increasing, $\zeta_i(0)=0$, and $\zeta_i(x)/x=\bigl(\xi_i(a_i+x)-a_i\bigr)/x$ is non-decreasing on $(0,\infty)$ by star-shapedness of $\xi_i$ at $a_i$. By Proposition~\ref{prop:VaR_X_convex}, the vector $(\zeta_1(X_1),\dots,\zeta_n(X_n))$ is VaR super-additive. Since $\widetilde X_i^{a_i}=a_i+\zeta_i(X_i)$, the translation identity \eqref{eq:gap-shift}, which rests only on Lemma~\ref{lm:VaRtoolkit}(ii), shows that $\widetilde{\XX}^{\aa}$ is VaR super-additive. If $\xi_i$ is convex with $\xi_i(a_i)=a_i$, then $\zeta_i$ is convex with $\zeta_i(0)=0$, hence star-shaped.

(ii) Here $\XX=\bb-\XX^{\bb}$, and NUOD of $\XX^{\bb}$ is NLOD of $\XX$, because $\{X_i^{b_i}\ge b_i-x_i\}=\{X_i\le x_i\}$. By continuity of the margins, $\phi_i=\phi_i^{b_i}$, so each $\phi_i$ is non-increasing. Put $\zeta_i(x)=b_i-\xi_i(b_i-x)$ for $x\ge0$. Then $\zeta_i$ is strictly increasing with $\zeta_i(0)=0$, and with $u=b_i-x$,
\[
\frac{\zeta_i(x)}{x}=\frac{b_i-\xi_i(u)}{b_i-u},
\]
which is non-increasing in $u$ by anti-star-shapedness, hence non-decreasing in $x$. By Proposition~\ref{prop:VaR_X_convex}, $(\zeta_1(X_1),\dots,\zeta_n(X_n))$ is VaR super-additive. Since $\widetilde X_i^{b_i}=b_i-\zeta_i(X_i)$, the reflection identity \eqref{eq:gap-reflect}, which rests only on Lemma~\ref{lm:VaRtoolkit}(iv), together with Lemma~\ref{lm:VaRversion}, shows that $\widetilde{\XX}^{\bb}$ is VaR sub-additive. If $\xi_i$ is concave with $\xi_i(b_i)=b_i$, then $x\mapsto\xi_i(b_i-x)$ is concave, so $\zeta_i$ is convex with $\zeta_i(0)=0$, hence star-shaped.
\end{proof}

\section{Computations for Section~\ref{sec:VaR_super-additivity}}
\label{app:margins}
\subsection{The laws of Table~\ref{tab:margins}}
\label{app:table}
Throughout, \(\phi_i'(x)=\log F_{X_i}(x)+xf_{X_i}(x)/F_{X_i}(x)\). Items (1)--(7) below correspond to the rows of Table~\ref{tab:margins}. Row~(8) follows from Remark~\ref{rm:phi_r_X_sufficicent} when \(\alpha_i\le1\), and from Proposition~\ref{prop:phi_mean_infinite} when \(\alpha_i>1\), since the mean is then finite.
\begin{itemize}

%-------------------------- 1 Frechet ------------------------------
\item[(1)] \textbf{Fr\'echet distribution.}  
The CDF is
\[
F_{X_i}(x_i)=\exp\left(-\left(\dfrac{x_i}{\theta_i}\right)^{-\alpha_i}\right),
\qquad x_i\ge0,\ \alpha_i,\theta_i>0,
\]
which yields
\begin{align*}
\phi_i(x_i)
&=-x_i\left(\frac{x_i}{\theta_i}\right)^{-\alpha_i}
\\
&=-\theta_i^{\alpha_i} x_i^{\,1-\alpha_i}.
\end{align*}
This function is non-increasing if and only if $x_i^{\,1-\alpha_i}$ is non-decreasing, that is, precisely when $0<\alpha_i\le 1$.

%-------------------------- 2 Lomax ------------------------------
\item[(2)] \textbf{Pareto(II)/Lomax distribution.}  
Here
\[
F_{X_i}(x_i)=1-\left(\dfrac{\theta_i}{\theta_i+x_i}\right)^{\alpha_i},
\qquad x_i\ge0,\ \alpha_i,\theta_i>0,
\]
and thus
\[
\phi_i(x_i)=x_i\log\!\left(1-\left(\dfrac{\theta_i}{\theta_i+x_i}\right)^{\alpha_i}\right).
\]
The derivative becomes
\begin{align*}
\phi_i^{'}(x_i)
&=\log F_{X_i}(x_i)+\dfrac{x_i f_{X_i}(x_i)}{F_{X_i}(x_i)}
\\
&=\log F_{X_i}(x_i)+\dfrac{\alpha_i x_i\Fbar_{X_i}(x_i)}{(\theta_i+x_i)F_{X_i}(x_i)}.
\end{align*}

\textbf{Claim.}  $\phi_i^{'}(x_i)\le 0$ for all $x_i\in[0,\infty)$ if and only if $0<\alpha_i\le 1$.

\smallskip
\emph{Necessity.}  
Assume $\phi_i^{'}(x_i)\le 0$ on $[0,\infty)$, and suppose $\alpha_i>1$. Using $\log(1-u)\sim-u$ as $u\to0$,
\[
\lim_{x_i\to\infty}
\dfrac{\phi_i^{'}(x_i)}{\Fbar_{X_i}(x_i)}
=\lim_{x_i\to\infty}\left[\frac{\log F_{X_i}(x_i)}{\Fbar_{X_i}(x_i)}+\frac{\alpha_i x_i}{(\theta_i+x_i)F_{X_i}(x_i)}\right]
=-1+\alpha_i.
\]
Since $\alpha_i>1$, the ratio becomes positive for sufficiently large $x_i$, contradicting the non-positivity of $\phi_i^{'}$. Thus, necessarily $0<\alpha_i\le 1$.

\smallskip
\emph{Sufficiency.}  
Assume $0<\alpha_i\le 1$ and abbreviate $F=F_{X_i}(x_i)$. Using $\log F=-\int_F^1 w^{-1}\,\mathrm{d}w$ and $1-F=\int_F^1\mathrm{d}w$, rewrite the derivative as
\begin{align*}
\phi_i^{'}(x_i)
&=-\int_{F}^1\dfrac{1}{w}\,\mathrm{d}w
+\dfrac{\alpha_i x_i}{(\theta_i+x_i)F}
\int_{F}^1\mathrm{d}w
\\
&=-\int_{F}^1
\dfrac{(\theta_i+x_i)F-\alpha_i x_i w}
{w(\theta_i+x_i)F}
\,\mathrm{d}w.
\end{align*}
Since $F\le w\le 1$, the integrand is non-negative as soon as it is non-negative at $w=1$, that is, as soon as 
\[
(\theta_i+x_i)F_{X_i}(x_i)-\alpha_i x_i\ge 0.
\]
Applying the mean value theorem to $u\mapsto u^{\alpha_i}$ on $\left[\dfrac{\theta_i}{\theta_i+x_i},1\right]$ gives
\[
F_{X_i}(x_i)=1-\left(\frac{\theta_i}{\theta_i+x_i}\right)^{\alpha_i}
=\alpha_i c^{\,\alpha_i-1}\,\frac{x_i}{\theta_i+x_i}
\;\ge\;\alpha_i\,\frac{x_i}{\theta_i+x_i}
\]
for some $c\in\left(\dfrac{\theta_i}{\theta_i+x_i},1\right)$, because $c<1$ and $\alpha_i-1\le0$ force $c^{\,\alpha_i-1}\ge1$. Multiplying by $\theta_i+x_i$ yields the required inequality, hence $\phi_i'\le0$.

Therefore, $\phi_i$ is non-increasing if and only if $0<\alpha_i\le 1$.

%-------------------------- 3 Levy ------------------------------
\item[(3)] \textbf{L\'evy distribution.}  
With
\[
F_{X_i}(x_i)=\mathrm{erfc}\!\left(\sqrt{\frac{\theta_i}{2x_i}}\right),\qquad x_i\ge0,\ \theta_i>0,
\]
define 
\[
\phi_i(x_i)=x_i\log\!\left(\mathrm{erfc}\!\left(\sqrt{\frac{\theta_i}{2x_i}}\right)\right).
\]
Differentiating gives
\[
\phi_i'(x_i)
=\log\!\left(\mathrm{erfc}\!\left(\sqrt{\frac{\theta_i}{2x_i}}\right)\right)
+\frac{\sqrt{\dfrac{\theta_i}{2\pi x_i}}
\exp\!\left(-\dfrac{\theta_i}{2x_i}\right)}
{\mathrm{erfc}\!\left(\sqrt{\dfrac{\theta_i}{2x_i}}\right)}.
\]

Let $\tau=\sqrt{\dfrac{\theta_i}{2x_i}}$, a decreasing bijection from $(0,\infty)$ onto itself.  
Then $\phi_i'(x_i)\le 0$ for all $x_i>0$ is equivalent to $\chi_i(\tau)\le 0$ for all $\tau>0$, where
\[
\chi_i(\tau)=\log(\mathrm{erfc}(\tau))
+\frac{\tau\exp(-\tau^2)}{\sqrt{\pi}\,\mathrm{erfc}(\tau)}.
\]
Since $\chi_i(0^+)=\log 1+0=0$, it suffices to show that $\chi_i'\le0$ on $(0,\infty)$. Differentiation leads to
\[
\chi_i'(\tau)
=\frac{\exp(-2\tau^2)
\left(2\tau-\sqrt{\pi}\,\exp(\tau^2)(2\tau^2+1)\mathrm{erfc}(\tau)\right)}
{\pi\,\mathrm{erfc}(\tau)^2},
\]
which is non-positive precisely when
\[
\frac{2\tau\exp(-\tau^2)}{\sqrt{\pi}\,(2\tau^2+1)}
\;\le\; \mathrm{erfc}(\tau).
\]
This is the lower bound of \citet{Gordon1941} for Mills' ratio, written for the complementary error function; for completeness, let
\[
g(\tau)=\mathrm{erfc}(\tau)-\frac{2\tau\exp(-\tau^2)}{\sqrt{\pi}\,(2\tau^2+1)} .
\]
Then $g(\tau)\to0$ as $\tau\to\infty$, and a direct computation gives
\[
g'(\tau)=-\frac{2\exp(-\tau^2)}{\sqrt{\pi}}\left[1+\frac{1-4\tau^2-4\tau^4}{(2\tau^2+1)^2}\right]=-\frac{4\exp(-\tau^2)}{\sqrt{\pi}\,(2\tau^2+1)^2}<0,
\]
so $g$ decreases to $0$ and is therefore non-negative on $(0,\infty)$.  
Thus $\chi_i\le0$, hence $\phi_i'(x_i)\le 0$ for all $x_i\ge0$, and $\phi_i$ is non-increasing on $[0,\infty)$ for every $\theta_i>0$.

%-------------------------- 4 Beta prime ------------------------------
\item[(4)] \textbf{One-parameter Beta Prime distribution.}  
With
\[
F_{X_i}(x_i)=\left(\dfrac{x_i}{1+x_i}\right)^{\alpha_i},\qquad x_i\ge0,\ \alpha_i>0,
\]
we have
\[
\phi_i(x_i)=\alpha_i x_i\log\!\left(\dfrac{x_i}{1+x_i}\right).
\]
Differentiation gives
\begin{align*}
\phi_i'(x_i)
&=\alpha_i\left(\log\!\left(\dfrac{x_i}{1+x_i}\right)
+\dfrac{1}{1+x_i}\right)
\\
&=\alpha_i\left(
-\sum_{k=1}^\infty\frac{1}{k}\left(\frac{1}{1+x_i}\right)^k
+\dfrac{1}{1+x_i}
\right)
\\
&\le \alpha_i\left(
-\frac{1}{1+x_i}+\dfrac{1}{1+x_i}
\right)=0,
\end{align*}
where the series expansion $\log(1-u)=-\sum_{k\ge1}u^k/k$ is applied with $u=(1+x_i)^{-1}\in(0,1)$ and all discarded terms are non-positive. Hence $\phi_i$ is non-increasing for every $\alpha_i>0$, with strict decrease on $(0,\infty)$.

%-------------------------- 5 Log-hazard ------------------------------
\item[(5)] \textbf{Log-Power distribution.}  
If
\[
F_{X_i}(x_i)
=\exp\!\left(-\dfrac{\bigl[\log(1+x_i)\bigr]^{\alpha_i}}{x_i}\right),
\qquad x_i\ge0,\ \alpha_i\in(-\infty,1),
\]
then
\[
\phi_i(x_i)=-\bigl[\log(1+x_i)\bigr]^{\alpha_i}.
\]
Two remarks are in order. First, the restriction $\alpha_i<1$ is one of admissibility rather than of monotonicity: since $\log(1+x)\sim x$ as $x\downarrow0$, the exponent behaves like $x^{\alpha_i-1}$, so $F_{X_i}(0^+)=0$ -- and hence Assumption~\ref{ass:zero-endpoint} -- holds exactly when $\alpha_i<1$, whereas $\alpha_i=1$ produces an atom of mass $e^{-1}$ at the origin, and for $\alpha_i>1$ the exponent vanishes at the origin, so that the formula decreases from one and does not define a distribution function. That $F_{X_i}$ is a genuine CDF for $\alpha_i<1$ follows from $(1+x)\log(1+x)\ge x$. Second, $\phi_i$ is non-increasing if and only if $\bigl[\log(1+x_i)\bigr]^{\alpha_i}$ is non-decreasing, which, $\log(1+\cdot)$ being an increasing bijection of $[0,\infty)$ onto itself, occurs exactly when $\alpha_i\ge0$. Combining the two constraints, $\phi_i$ is non-increasing precisely for $0\le\alpha_i<1$; the boundary case $\alpha_i=0$ returns the Fr\'echet law of part~(1) with unit shape.

%-------------------------- 6 Log-Cauchy ------------------------------
\item[(6)] \textbf{Log-Cauchy distribution.}  
The CDF is
\[
F_{X_i}(x_i)=\dfrac{1}{2}+\dfrac{1}{\pi}\arctan(\alpha_i\log(x_i)),
\qquad x_i> 0,\ \alpha_i>0.
\]
Hence
\[
\phi_i(x_i)
=x_i\log\!\left(\dfrac{1}{2}
+\dfrac{1}{\pi}\arctan(\alpha_i\log(x_i))\right).
\]
Differentiation yields
\[
\phi_i'(x_i)
=\log F_{X_i}(x_i)
+\dfrac{\alpha_i}{\pi\bigl(1+(\alpha_i\log(x_i))^2\bigr)\,F_{X_i}(x_i)}.
\]
Introducing $\vartheta=\arctan(\alpha_i\log(x_i))$, so that $1+(\alpha_i\log x_i)^2=\sec^2\vartheta$ and $\vartheta=\pi\bigl(F_{X_i}(x_i)-\tfrac12\bigr)$, gives
\[
\phi_i'(x_i)
=\log F_{X_i}(x_i)
+\dfrac{\alpha_i \cos^2\vartheta}{\pi F_{X_i}(x_i)}
=\log F_{X_i}(x_i)
+\dfrac{\alpha_i\sin^2\bigl(\pi F_{X_i}(x_i)\bigr)}{\pi F_{X_i}(x_i)}.
\]

Since $x_i\mapsto F_{X_i}(x_i)$ is a bijection from $(0,\infty)$ onto $(0,1)$, non-positivity of $\phi_i'$ on $(0,\infty)$ is equivalent to $\chi_i(u)\le0$ for all $u\in(0,1)$, where
\[
\chi_i(u)=\log(u)
+\dfrac{\alpha_i\sin^2(\pi u)}{\pi u}.
\]
Rearranging, this holds if and only if
\[
0<\alpha_i\le \alpha^{\star}:=\inf_{u\in(0,1)}
\dfrac{-\pi u\log(u)}{\sin^2(\pi u)}.
\]
The ratio tends to $+\infty$ at both endpoints and attains an interior minimum near $u\approx0.5672$, whence $\alpha^{\star}\approx1.0568$. Thus $\phi_i$ is non-increasing on $(0,\infty)$ if and only if $0<\alpha_i\le\alpha^{\star}$; notably the threshold exceeds one, so a range of log-Cauchy laws strictly beyond the ``shape at most one'' boundary of parts~(1) and~(2) is admissible.

%-------------------------- 7 Inverse gamma ------------------------------
\item[(7)] \textbf{Inverse-Gamma distribution.}  
The CDF can be written as
\[
F_{X_i}(x_i)
=\frac{1}{\Gamma(\alpha_i)}\Gamma\!\left(\alpha_i,\dfrac{\theta_i}{x_i}\right),
\qquad x_i\ge0,\ \alpha_i,\theta_i>0,
\]
with $\Gamma(\cdot,\cdot)$ the upper incomplete gamma function, leading to
\[
\phi_i(x_i)
=x_i\log\!\left(\frac{1}{\Gamma(\alpha_i)}
\Gamma\!\left(\alpha_i,\frac{\theta_i}{x_i}\right)\right).
\]
Differentiation gives
\[
\phi_i'(x_i)
=\log F_{X_i}(x_i)
+\dfrac{\left(\dfrac{\theta_i}{x_i}\right)^{\alpha_i}
\exp(-\theta_i/x_i)}
{\Gamma(\alpha_i)\,F_{X_i}(x_i)}.
\]

Let $\tau=\theta_i/x_i$.  
Then $\phi_i'(x_i)\le 0$ is equivalent to $\chi_i(\tau)\le 0$, where
\[
\chi_i(\tau)
=\log\!\left(\frac{\Gamma(\alpha_i,\tau)}{\Gamma(\alpha_i)}\right)
+\dfrac{\tau^{\alpha_i}\exp(-\tau)}{\Gamma(\alpha_i,\tau)}.
\]

\textbf{Claim.}  $\chi_i(\tau)\le 0$ for all $\tau\ge0$ if and only if $0<\alpha_i\le 1$.

\smallskip
\emph{Necessity.}  
Using $\Gamma(\alpha_i,\tau)\sim \tau^{\alpha_i-1}e^{-\tau}$ as $\tau\to\infty$, the two terms of $\chi_i$ equal $-\tau+(\alpha_i-1)\log \tau-\log\Gamma(\alpha_i)+o(1)$ and $\tau-(\alpha_i-1)+o(1)$ respectively, so that
\[
\lim_{\tau\to0^+}\chi_i(\tau)=0,\qquad
\lim_{\tau\to\infty}\chi_i(\tau)=
\begin{cases}
-\infty, & 0<\alpha_i<1,\\
0, & \alpha_i=1,\\
+\infty, & \alpha_i>1,
\end{cases}
\]
so non-positivity requires $0<\alpha_i\le1$.

\smallskip
\emph{Sufficiency.}  
If $0<\alpha_i\le1$, then
\[
\chi_i'(\tau)
=\frac{
\tau^{\alpha_i-1}\exp(-2\tau)
\bigl(\tau^{\alpha_i}-\exp(\tau)(\tau+1-\alpha_i)\Gamma(\alpha_i,\tau)\bigr)
}
{\Gamma(\alpha_i,\tau)^2},
\]
which is non-positive precisely when
\[
\frac{\tau^{\alpha_i}\exp(-\tau)}{\tau+1-\alpha_i}
\le \Gamma(\alpha_i,\tau),
\]
the denominator being positive because $\alpha_i\le1$. To verify it, let $h(\tau)=\Gamma(\alpha_i,\tau)-\tau^{\alpha_i}e^{-\tau}/(\tau+1-\alpha_i)$. Then $h(\tau)\to0$ as $\tau\to\infty$, and a direct computation gives
\[
h'(\tau)=-\frac{(1-\alpha_i)\,\tau^{\alpha_i-1}e^{-\tau}}{(\tau+1-\alpha_i)^2}\ \le\ 0,
\]
so $h$ decreases to $0$ and is non-negative; at $\alpha_i=1$ the bound is an identity.  
Thus $\chi_i$ is non-increasing with $\chi_i(0^+)=0$, proving $\chi_i(\tau)\le 0$ on $[0,\infty)$. The boundary case $\alpha_i=1$ returns the Fr\'echet law of part~(1) with unit shape, for which $\phi_i\equiv-\theta_i$ is constant.

\smallskip
Therefore, $\phi_i$ is non-increasing on $[0,\infty)$ if and only if $0<\alpha_i\le 1$.

\end{itemize}

\subsection{Simplex dominance in Example~\ref{example:SDnotNI}}
\label{app:SDnotNI}
Let \(t=x_1+x_2+x_3\) and \(R=\Phi(x_1,x_2,x_3)-\Phi(t,t,t)=-\sqrt{x_1}-\sqrt{x_2}+\phi_3(x_3)+2\sqrt t-\phi_3(t)\). On \([0,1]\), \(\phi_3'(x)=2\log x+1\), so \(\phi_3\) decreases on \([0,e^{-1/2}]\), increases on \([e^{-1/2},1]\), and satisfies \(\min_{[0,1]}\phi_3=-2e^{-1/2}\) and \(\phi_3(1)=-1\). Put \(\gamma(x)=\phi_3(x)+\sqrt x\); then \(\gamma\equiv0\) on \([1,\infty)\), and \(\gamma(x)\ge x(2\log x-1)\ge-2e^{-1/2}\) on \([0,1]\).

\emph{Case \(t\ge1\).} Here \(\phi_3(t)=-\sqrt t\), and \(\sqrt{x_1}+\sqrt{x_2}+\sqrt{x_3}\le\sqrt{3t}\) by the Cauchy--Schwarz inequality. Hence
\[
R=3\sqrt t-\sum_{i=1}^3\sqrt{x_i}+\gamma(x_3)\ \ge\ (3-\sqrt3)\sqrt t-2e^{-1/2}\ \ge\ 3-\sqrt3-2e^{-1/2}>0,
\]
because \(3-\sqrt3\approx1.2679\) exceeds \(2e^{-1/2}\approx1.2131\).

\emph{Case \(t\le e^{-1/2}\).} Now \(x_3\le t\le e^{-1/2}\), and \(\phi_3\) is decreasing on \([0,e^{-1/2}]\), so \(\phi_3(x_3)\ge\phi_3(t)\); moreover \(\sqrt{x_1}+\sqrt{x_2}\le\sqrt{2(x_1+x_2)}\le\sqrt{2t}\). Hence \(R\ge(2-\sqrt2)\sqrt t\ge0\).

\emph{Case \(e^{-1/2}<t<1\).} Now \(\phi_3(x_3)\ge-2e^{-1/2}\), \(\phi_3(t)\le\phi_3(1)=-1\) because \(\phi_3\) is increasing on \([e^{-1/2},1]\), and again \(\sqrt{x_1}+\sqrt{x_2}\le\sqrt{2t}\). Hence
\[
R\ \ge\ (2-\sqrt2)\sqrt t+1-2e^{-1/2}\ \ge\ (2-\sqrt2)e^{-1/4}+1-2e^{-1/2}\approx0.243>0.
\]
In all cases \(R\ge0\), so \(\Phi\) is SD.
\end{appendices}

\begin{thebibliography}{}
\bibitem[Acerbi and Tasche, 2002]{Acerbi.2002}
Acerbi, C. and Tasche, D. (2002).
\newblock Expected shortfall: a natural coherent alternative to {V}alue at {R}isk.
\newblock {\em Economic Notes}, 31(2):379--388.

\bibitem[Arab et~al., 2025]{Arab2025}
Arab, I., Lando, T., and Oliveira, P.~E. (2025).
\newblock Convex combinations of random variables stochastically dominate the
  parent for a new class of heavy tailed distributions.
\newblock {\em Electronic Communications in Probability}, 30:1--11.

\bibitem[Artzner et~al., 1999]{Artzner1999}
Artzner, P., Delbaen, F., Eber, J.-M., and Heath, D. (1999).
\newblock Coherent measures of risk.
\newblock {\em Mathematical Finance}, 9(3):203--228.

\bibitem[Block et~al., 1982]{Block1982a}
Block, H.~W., Savits, T.~H., and Shaked, M. (1982).
\newblock Some concepts of negative dependence.
\newblock {\em The Annals of Probability}, 10(3):765--772.

\bibitem[Block et~al., 1985]{Block1985}
Block, H.~W., Savits, T.~H., and Shaked, M. (1985).
\newblock A concept of negative dependence using stochastic ordering.
\newblock {\em Statistics \& Probability Letters}, 3(2):81--86.

\bibitem[Block et~al., 1998]{Block1998}
Block, H.~W., Savits, T.~H., and Singh, H. (1998).
\newblock The reversed hazard rate function.
\newblock {\em Probability in the Engineering and Informational Sciences},
  12(1):69--90.

\bibitem[Chen et~al., 2025a]{Chen2025}
Chen, Y., Embrechts, P., and Wang, R. (2025).
\newblock An unexpected stochastic dominance: {P}areto distributions,
  dependence, and diversification.
\newblock {\em Operations Research}, 73(3):1336--1344.

\bibitem[Chen et~al., 2025b]{Chen2025a}
Chen, Y., Hu, T., Shneer, S., and Zou, Z. (2025).
\newblock Stochastic dominance for linear combinations of infinite-mean risks.
\newblock {\em arXiv preprint} arXiv:2505.01739.

\bibitem[Chen et~al., 2025c]{Chen2025b}
Chen, Y., Embrechts, P., and Wang, R. (2025).
\newblock Risk exchange under infinite-mean {P}areto models.
\newblock {\em Insurance: Mathematics and Economics}, 124:103131.

\bibitem[Chen and Shneer, 2026]{Chen2026}
Chen, Y. and Shneer, S. (2026).
\newblock Risk aggregation and stochastic dominance for a class of heavy-tailed
  distributions.
\newblock {\em ASTIN Bulletin}, 56(1):206--219.

\bibitem[Dan{\'\i}elsson et~al., 2013]{Danielsson2013}
Dan{\'\i}elsson, J., Jorgensen, B.~N., Samorodnitsky, G., Sarma, M., and de~Vries,
  C.~G. (2013).
\newblock Fat tails, {VaR} and subadditivity.
\newblock {\em Journal of Econometrics}, 172(2):283--291.

\bibitem[de~Haan, 1984]{deHaan1984}
de~Haan, L. (1984).
\newblock A spectral representation for max-stable processes.
\newblock {\em The Annals of Probability}, 12(4):1194--1204.

\bibitem[Denuit et~al., 2005]{Denuit2005}
Denuit, M., Dhaene, J., Goovaerts, M., and Kaas, R. (2005).
\newblock {\em Actuarial Theory for Dependent Risks: Measures, Orders and
  Models}.
\newblock Wiley, Chichester.

\bibitem[Dhaene and Denuit, 1999]{Dhaene1999}
Dhaene, J. and Denuit, M. (1999).
\newblock The safest dependence structure among risks.
\newblock {\em Insurance: Mathematics and Economics}, 25(1):11--21.

\bibitem[Dhaene et~al., 2002]{Dhaene2002}
Dhaene, J., Denuit, M., Goovaerts, M.~J., Kaas, R., and Vyncke, D. (2002).
\newblock The concept of comonotonicity in actuarial science and finance:
  theory.
\newblock {\em Insurance: Mathematics and Economics}, 31(1):3--33.

\bibitem[Embrechts and Goldie, 1980]{Embrechts1980}
Embrechts, P. and Goldie, C.~M. (1980).
\newblock On closure and factorization properties of subexponential and
  related distributions.
\newblock {\em Journal of the Australian Mathematical Society}, 29(2):243--256.

\bibitem[Embrechts and Hofert, 2013]{Embrechts2013}
Embrechts, P. and Hofert, M. (2013).
\newblock A note on generalized inverses.
\newblock {\em Mathematical Methods of Operations Research}, 77(3):423--432.

\bibitem[Embrechts et~al., 1997]{Embrechts1997}
Embrechts, P., Kl{\"u}ppelberg, C., and Mikosch, T. (1997).
\newblock {\em Modelling Extremal Events for Insurance and Finance}.
\newblock Springer, Berlin.

\bibitem[Embrechts et~al., 1999]{Embrechts1999}
Embrechts, P., Resnick, S.~I., and Samorodnitsky, G. (1999).
\newblock Extreme value theory as a risk management tool.
\newblock {\em North American Actuarial Journal}, 3(2):30--41.

\bibitem[Embrechts et~al., 2009a]{Embrechts2008}
Embrechts, P., Lambrigger, D.~D., and W{\"u}thrich, M.~V. (2009).
\newblock Multivariate extremes and the aggregation of dependent risks:
  examples and counter-examples.
\newblock {\em Extremes}, 12(2):107--127.

\bibitem[Embrechts et~al., 2009b]{Embrechts2009b}
Embrechts, P., Ne{\v s}lehov{\'a}, J., and W{\"u}thrich, M.~V. (2009).
\newblock Additivity properties for {V}alue-at-{R}isk under {A}rchimedean
  dependence and heavy-tailedness.
\newblock {\em Insurance: Mathematics and Economics}, 44(2):164--169.

\bibitem[Embrechts et~al., 2014]{Embrechts2014}
Embrechts, P., Puccetti, G., R{\"u}schendorf, L., Wang, R., and Beleraj, A.
  (2014).
\newblock An academic response to {B}asel 3.5.
\newblock {\em Risks}, 2(1):25--48.

\bibitem[Gordon, 1941]{Gordon1941}
Gordon, R.~D. (1941).
\newblock Values of {M}ills' ratio of area to bounding ordinate and of the
  normal probability integral for large values of the argument.
\newblock {\em The Annals of Mathematical Statistics}, 12(3):364--366.

\bibitem[Ibragimov, 2009]{Ibragimov2009}
Ibragimov, R. (2009).
\newblock Portfolio diversification and value at risk under thick-tailedness.
\newblock {\em Quantitative Finance}, 9(5):565--580.

\bibitem[Ibragimov and Walden, 2008]{Ibragimov2008}
Ibragimov, R. and Walden, J. (2008).
\newblock Portfolio diversification under local and moderate deviations from
  power laws.
\newblock {\em Insurance: Mathematics and Economics}, 42(2):594--599.

\bibitem[Imamura and Kato, 2026]{Imamura2025}
Imamura, Y. and Kato, T. (2026).
\newblock A note on subadditivity of value at risks ({VaRs}): a new connection
  to comonotonicity.
\newblock {\em Journal of Applied Probability}, 63(1):91--95.

\bibitem[Jobst, 2007]{Jobst2007}
Jobst, A. (2007).
\newblock {\em Operational Risk: The Sting is Still in the Tail but the Poison
  Depends on the Dose}.
\newblock International Monetary Fund, Washington, DC.

\bibitem[Joe, 1997]{Joe1997}
Joe, H. (1997).
\newblock {\em Multivariate Models and Dependence Concepts}.
\newblock Chapman and Hall, London.

\bibitem[Klugman et~al., 2019]{Klugman2019}
Klugman, S.~A., Panjer, H.~H., and Willmot, G.~E. (2019).
\newblock {\em Loss Models: From Data to Decisions}.
\newblock Wiley, Hoboken, NJ, 5th edition.

\bibitem[Lai and Xie, 2006]{Lai2006}
Lai, C.~D. and Xie, M. (2006).
\newblock {\em Stochastic Ageing and Dependence for Reliability}.
\newblock Springer, New York.

\bibitem[Lindskog and McNeil, 2003]{Lindskog2003}
Lindskog, F. and McNeil, A.~J. (2003).
\newblock Common {P}oisson shock models: applications to insurance and credit
  risk modelling.
\newblock {\em ASTIN Bulletin}, 33(2):209--238.

\bibitem[Linsmeier and Pearson, 2000]{Linsmeier2000}
Linsmeier, T.~J. and Pearson, N.~D. (2000).
\newblock Value at risk.
\newblock {\em Financial Analysts Journal}, 56(2):47--67.

\bibitem[Marshall and Olkin, 1967]{Marshall1967}
Marshall, A.~W. and Olkin, I. (1967).
\newblock A multivariate exponential distribution.
\newblock {\em Journal of the American Statistical Association}, 62(317):30--44.

\bibitem[McNeil et~al., 2015]{McNeil2015}
McNeil, A.~J., Frey, R., and Embrechts, P. (2015).
\newblock {\em Quantitative Risk Management: Concepts, Techniques and
  Tools}.
\newblock Princeton University Press, Princeton, revised edition.

\bibitem[M{\"u}ller, 2025]{Mueller2025}
M{\"u}ller, A. (2025).
\newblock Some remarks on the effect of risk sharing and diversification for
  infinite mean risks.
\newblock {\em ASTIN Bulletin}, 55(3):747--756.

\bibitem[Nelsen, 2006]{10.5555/1952073}
Nelsen, R.~B. (2006).
\newblock {\em An Introduction to Copulas}.
\newblock Springer, New York, 2nd edition.

\bibitem[Puccetti et~al., 2013]{Puccetti2013}
Puccetti, G., Wang, B., and Wang, R. (2013).
\newblock Complete mixability and asymptotic equivalence of worst-possible
  {VaR} and {ES} estimates.
\newblock {\em Insurance: Mathematics and Economics}, 53(3):821--828.

\bibitem[R{\"u}schendorf, 2013]{Ruschendorf2013}
R{\"u}schendorf, L. (2013).
\newblock {\em Mathematical Risk Analysis}.
\newblock Springer, Heidelberg.

\bibitem[Shaked and Shanthikumar, 2007]{Shaked2007}
Shaked, M. and Shanthikumar, J.~G. (2007).
\newblock {\em Stochastic Orders}.
\newblock Springer, New York.

\bibitem[Tasche, 2002]{Tasche2002}
Tasche, D. (2002).
\newblock Expected shortfall and beyond.
\newblock {\em Journal of Banking \& Finance}, 26(7):1519--1533.

\bibitem[Vincent, 2025]{Vincent2025}
Vincent, L. (2025).
\newblock Diversification and stochastic dominance: when all eggs are better put
  in one basket.
\newblock {\em arXiv preprint} arXiv:2507.16265.

\bibitem[Wang and Zitikis, 2021]{Wang2021}
Wang, R. and Zitikis, R. (2021).
\newblock An axiomatic foundation for the expected shortfall.
\newblock {\em Management Science}, 67(3):1413--1429.

\bibitem[Zeng et~al., 2025]{Zeng2025}
Zeng, K., Zou, Z., Su, Y., and Hu, T. (2025).
\newblock Further developments on stochastic dominance for convex combinations
  of infinite-mean random variables.
\newblock {\em arXiv preprint} arXiv:2511.00764.

\bibitem[Zhu et~al., 2023]{Zhu2023}
Zhu, W., Li, L., Yang, J., Xie, J., and Sun, L. (2023).
\newblock Asymptotic subadditivity/superadditivity of {V}alue-at-{R}isk under
  tail dependence.
\newblock {\em Mathematical Finance}, 33(4):1314--1369.

\end{thebibliography}
\end{document}